\pdfoutput=1
\documentclass[journal]{IEEEtran}
\usepackage{cite}
\usepackage{amsmath}
\usepackage{amssymb}
\usepackage{amsthm}
\usepackage{arydshln}
\usepackage{mathtools}
\usepackage{mathrsfs}
\usepackage{algorithm}
\usepackage{algorithmic}
\usepackage{graphicx}
\usepackage{epsfig,graphics}
\usepackage{epstopdf}
\usepackage{subfigure}
\usepackage{enumitem}
\usepackage{pifont}
\usepackage{hyperref}
\usepackage{multirow,ragged2e}
\usepackage{bm}
\usepackage{fixltx2e}
\usepackage[english]{babel}
\usepackage{color}
\usepackage{xcolor}

\usepackage{filecontents}
\theoremstyle{theorem}
\newtheorem{theorem}{Theorem}
\theoremstyle{lemma}
\newtheorem{lemma}{Lemma}
\theoremstyle{definition}
\newtheorem{definition}{Definition}
\theoremstyle{assumption}
\newtheorem{assumption}{Assumption}
\theoremstyle{problem}
\newtheorem{problem}{Problem}
\theoremstyle{example}

\theoremstyle{proposition}
\newtheorem{proposition}{Proposition}
\theoremstyle{corollary}
\newtheorem{corollary}{Corollary}

\newtheorem{remark}{Remark}

\newcommand{\Input}{\textbf{Input:} $r$, $\epsilon$, $\overline{f}$, $\lambda_{i2}$, $c_{i}$, $\overline{T}_{ij}$, $\|\fbm{x}^{d}_{ij}\|\ \forall j\in\mathcal{N}_{i}(0)$}
\newcommand{\Output}{\textbf{Output:} $K_{i}(t)$, $B$ and $D$}

\newcommand{\fig}[1]{Figure~\ref{#1}}

\newcommand{\sect}[1]{Section~\ref{#1}}

\newcommand{\eq}[1]{Equation~(\ref{#1})}

\newcommand{\ass}[1]{Assumption~\ref{#1}}
\newcommand{\asses}[2]{Assumptions~\ref{#1}-\ref{#2}}
\newcommand{\prop}[1]{Proposition~\ref{#1}}
\newcommand{\lem}[1]{Lemma~\ref{#1}}
\newcommand{\rem}[1]{Remark~\ref{#1}}
\newcommand{\prob}[1]{Problem~\ref{#1}}
\newcommand{\theo}[1]{Theorem~\ref{#1}}
\newcommand{\defi}[1]{Definition~\ref{#1}}

\newcommand{\eqrefs}[2]{(\ref{#1})-(\ref{#2})}

\newcommand{\fbm}[1]{\mathbf{#1}}
\newcommand{\dbm}[1]{\dot{\fbm{#1}}}
\newcommand{\ddbm}[1]{\ddot{\fbm{#1}}}

\newcommand{\tbm}[1]{\fbm{#1}^\mathsf{T}}

\newcommand{\tfbm}[1]{\bm{#1}^\mathsf{T}}
\newcommand{\tilbm}[1]{\tilde{\fbm{#1}}}

\newcommand{\ttilbm}[1]{\tilde{\fbm{#1}}^\mathsf{T}}
\newcommand{\dtbm}[1]{\dot{\fbm{#1}}^\mathsf{T}}

\newcommand{\hfbm}[1]{\hat{\bm{#1}}}
\newcommand{\dfbm}[1]{\dot{\bm{#1}}}
\newcommand{\bbm}[1]{\overline{\fbm{#1}}}
\newcommand{\tbbm}[1]{\overline{\fbm{#1}}^\mathsf{T}}
\newcommand{\tilfbm}[1]{\tilde{\bm{#1}}}

\newcommand{\Sat}{\text{Sat}}

\begin{document}

\title{Connectivity-Preserving Swarm Teleoperation Over A Tree Network With Time-Varying Delays}

\author{Yuan~Yang, Yang Shi,~\IEEEmembership{Fellow,~IEEE} and Daniela Constantinescu,~\IEEEmembership{Member,~IEEE}\thanks{The authors are with the Department of Mechanical Engineering, University of Victoria, Victoria, BC V8W 2Y2 Canada (e-mail: yangyuan@uvic.ca; yshi@uvic.ca; danielac@uvic.ca).}}

\maketitle

%%%%%%%%%%%%%%%%%%%%%%%%%%%%%%%%%%%%%%%
\begin{abstract}
%%%%%%%%%%%%%%%%%%%%%%%%%%%%%%%%%%%%%%%
 
A teleoperated swarm must follow the unpredictable commands of its human operator while remaining connected. When the swarm communications are limited by distance and affected by delays, both the user input and the transmission delays endanger the connectivity of the swarm. This paper presents a constructive control strategy that overcomes both threats. The strategy modulates the intra-swarm couplings and the damping injected to each slave in the swarm based on a customized potential. Lyapunov-based set invariance analysis proves that the proposed explicit gain updating law limits the impact of the operator input and preserves the initial tree connectivity of a delay-free swarm. Further augmentation with stricter selection of control gains robustifies the design to time-varying delays in intra-swarm communications. The paper also establishes the input-to-state stability of a teleoperated time-delay swarm under the proposed dynamic control. Experiments validate connectivity maintenance and synchronization during time-delay swarm teleoperation with the proposed control.

\end{abstract}

\IEEEpeerreviewmaketitle

%%%%%%%%%%%%%%%%%%%%%%%%%%%%%%%%%%%%%%%
\section{Introduction}\label{sec:introduction}
%%%%%%%%%%%%%%%%%%%%%%%%%%%%%%%%%%%%%%%
Fully autonomous multi-robot systems~(MRS-s) have been extensively studied because they are more robust and flexible than single-robot systems~\cite{Schwager2011Proceedings, Kumar2011IJRR}. Compared to autonomous MRS-s, semi-autonomous teleoperated swarms are better suited for complex interactions in unpredictable environments~\cite{Paolo2012RAM} because human operators partially control them. The master may be connected to the swarm through velocity-like variables~\cite{Paolo2011RSS, Paolo2012TRO}, virtual kinematic points~\cite{Lee2013TMECH}, or a virtual rigid body~\cite{Schwager2016ICRA}, and may steer bearing formations~\cite{Paolo2012IJRR}, fly an arbitrary number of aerial robots with collision avoidance~\cite{Schwager2016ICRA}, or exchange forces between a human and a system of unmanned aerial vehicles~\cite{Heinrich2017IJRR}. Different haptic cues can improve maneuverability and perceptual sensitivity~\cite{Paolo2013TMECH}, while time-varying density functions can increase the agility of human-swarm adaptive interaction for optimal coverage control~\cite{Egerstedt2015TRO}. 

Because most robots have limited communication range, the distributed synchronization of MRS-s must also guarantee their connectivity. Most existing work preserves the connectivity of fully autonomous MRS-s during coordination. In particular, decentralized algebraic connectivity estimation~\cite{Sabattini2013IJRR, Paolo2013IJRR, Sabattini2013TRO} helps maintain global connectivity. For kinematic MRS-s, gradient-based controllers derived from unbounded~\cite{Egerstedt2007TRO, Dimos2007TAC, Dimos2008TRO, Zavlanos2008TRO} or bounded~\cite{Amir2010TAC, Amir2013TAC, Dixon2015TCNS} potentials preserve the local connectivity and account for disturbances~\cite{Dong2016Automatica}, Lipschitz nonlinearities~\cite{Ren2016TAC} and obstacles~\cite{Dixon2012TAC, Dimos2017TAC}. Topics of recent interest include intermittent connectivity~\cite{Hollinger2012TRO, Zavlanos2017TAC}, strong connectivity in directed graphs~\cite{Spong2017TAC}, robustness and invariance of connectivity maintenance with additional control terms~\cite{Dimos2017SIAM}, and the trade-offs among bounded controls, connectivity maintenance and additional control objectives~\cite{Sabattini2017TRO}. For dynamic MRS-s, the minimization of the total energy yields controls that preserve local connectivity during the sychronization of double-integrators~\cite{Su2010SCL, Su2011Automatica}, and during the leader-follower coordination of double integrators~\cite{Dong2013Automatica, Dong2014TAC, Su2015Automatica, Ai2016Automatica} and Euler-Lagrange~(EL) networks~\cite{Ren2012SCL, Dong2017IJRNC}. Robust controllers suppress disturbances in double-integrator rendezvous~\cite{Hu2017TCNS}, flocking~\cite{Dong2015Automatica} and formation tracking~\cite{Hu2018Automatica}.

For semi-autonomous MRS-s, passivity-based strategies maintain both global~\cite{Paolo2011RSS, Paolo2013ICRA} and local~\cite{Lee2013TMECH} connectivity. Global connectivity relies on gradient-based controls provided by a potential of the algebraic connectivity estimated by~\cite{Freeman2010Automatica}. Hence, the accuracy and rate of convergence of the connectivity estimates impact the safety of the teleoperated swarm~\cite{Dimos2014Automatica, Sabattini2015Cybernetics}. As discussed in~\cite{Sabattini2017TRO}, the connectivity estimation errors foil the ability to mitigate perturbations of global connectivity control even for first-order MRS-s. Local connectivity relies on gradient-based controls provided by a potential of displacements. Displacements between virtual kinematic points guarantee only the local connectivity of the virtual system~\cite{Lee2013TMECH}. The  control Lyapunov function validating the feasibility of the human-swarm interaction controls in~\cite{Egerstedt2015ACC} provides no explicit connectivity-preserving swarm teleoperation control strategy. 

A key feature of teleoperated swarms is that they are driven by human operators and, thus, are semi-autonomous. As external commands, the user inputs differ from the disturbances of autonomous systems and require a novel treatment in control. This paper proposes a new methodology to approach the user forces, which retains their driving role to the extent to which they do not endanger the connectivity of the swarm.

Without loss of generality, the paper assumes only one informed slave that has been passively connected to the master regardless of their distance~\cite{Lee2013TMECH}. It then focuses on maintaining all interaction links in the delay-free tree network of the slave swarm commanded by the user. Set invariance analysis shows that the connectivity of the swarm can be maintained through proper regulation of the energy of the swarm, as quantified by a customized potential of inter-robot distances. Structural controllability renders this energy provably upper-bounded by the local information of all slaves. Then, a transformation of the delay-free swarm dynamics into a first-order representation with state-dependent mismatches facilitates the introduction of a dynamic control that distributively regulates the energy of the teleoperated swarm. The dynamic control dominates the state-dependent mismatches and limits the user-injected energy by modulating the couplings between, and the damping injected to, slave robots according to their distances. Lyapunov energy analysis provably confines the impact of the unpredictable user command on the connectivity of the slave swarm to a safe domain.  

An augmentation of the proposed dynamic strategy then overcomes the threat of time-varying delays in the intra-swarm communications through stricter selection of the control gains. In time-delay swarms, the slaves move during the time interval from when they send out their position signals to when their neighbours receive them. To account for these delay-induced position mismatches, the augmented strategy (i) imposes stricter constraints on the inter-slave distances and (ii) develops a Lyapunov-Krasovskii functional to measure the impact on connectivity both of the user command and of the time-varying transmission delays. Set invariance analysis proves that the connectivity of the time-delay slave swarm can be guaranteed by the augmented dynamic control with certain conditions. An algorithm then verifies the feasibility of these conditions and, with them, of the proposed controller design. To the authors' best knowledge, the proposed dynamic strategy is the first to rigorously establish connectivity-preserving teleoperation of a time-delay swarm based on the structural controllability of a tree network. 

Lastly, the paper establishes input-to-state stability (ISS) of the time-delay teleoperated swarm. Because the model of the human operator is unknown, their command transmitted from the master side is unpredictable. Therefore, the proposed dynamic control stabilizes the time-delay teleoperated swarm subsystem with the transmitted user command as the input. The velocities of, and position errors between, the slaves in the tree network define the state of the slave swarm. Then, ISS swarm teleoperation provides an invariant set and a globally attractive set for measuring the robust position synchronization of the slave swarm under the perturbation of the transmitted user input. Specifically, the proposed dynamic strategy: (i) estimates the impact of the operator on the stability of the slave swarm using local position errors; (ii) updates the control gains to confine this impact to a domain safe for the robust position synchronization of the slave swarm; and (iii) maximally applies the operator command to the swarm. Experiments illustrate that, compared to the design in~\cite{Lee2013TMECH}, the proposed dynamic strategy can preserve the connectivity of teleoperated swarms with time-delayed information transmissions.

Preliminary results of some of the material from the paper has been reported in~\cite{Yuan_IROS2019}. Beyond the conference submission, this paper aims to tackle the threats to the swarm connectivity posed by not only the user perturbation but also the time-varying transmission delays. More specifically, the paper has four key contributions: (i) it reformulates the connectivity preservation problem of swarm teleoperation with time-varying transmission delays; (ii) it then augments the initial strategy with feasibility analysis to stabilize and to preserve the connectivity of a time-delay slave swarm; (iii) it further establishes the input-to-state stability of a teleoperated time-delay swarm with the proposed control; and (iv) it also validates, through experiments, that the new strategy enables the operator to guide the motion of a swarm while preventing time-varying delays from endangering swarm connectivity.

%%%%%%%%%%%%%%%%%%%%%%%%%%%%%%%%%%%%%%%
\section{Problem Formulation}\label{sec: problem formulation}
%%%%%%%%%%%%%%%%%%%%%%%%%%%%%%%%%%%%%%%
Let a swarm teleoperation system have one master robot and $N$ slave robots. An operator commands the group of slaves to a desired location by operating the master. To simplify the design without loss of generality, let the swarm have one informed slave, passively connected to the master robot regardless of their distance~\cite{Lee2013TMECH}. Then, this paper will preserve the connectivity of the slave swarm under the user command transmitted through the master-informed slave connection. 

Let the teleoperated slave swarm be a network of $n$-degree-of-freedom~($n$-DOF) Euler-Lagrange~(EL) systems: 
\begin{equation}\label{equ1}
\begin{aligned}
\fbm{M}_{1}(\fbm{x}_{1})\ddbm{x}_{1}+\fbm{C}_{1}(\fbm{x}_{1},\dbm{x}_{1})\dbm{x}_{1}=&\fbm{u}_{1}+\fbm{f}\textrm{,}\\
\fbm{M}_{s}(\fbm{x}_{s})\ddbm{x}_{s}+\fbm{C}_{s}(\fbm{x}_{s},\dbm{x}_{s})\dbm{x}_{s}=&\fbm{u}_{s}\textrm{,}
\end{aligned}
\end{equation}
where: the subscript $1$ indicates the informed slave that receives the time-varying user command $\fbm{f}$ from the master; and $s=2,\cdots,N$ index the remaining $N-1$ slaves. For each slave $i=1,\cdots,N$: $\fbm{x}_{i}$, $\dbm{x}_{i}$ and $\ddbm{x}_{i}$ are its position, velocity and acceleration vectors; $\fbm{M}_{i}(\fbm{x}_{i})$ and $\fbm{C}_{i}(\fbm{x}_{i},\dbm{x}_{i})$ are its matrices of inertia and of Coriolis and centrifugal effects; and $\fbm{u}_{i}$ is its control force to be designed. The EL dynamics~\eqref{equ1} have the following properties:
\begin{enumerate}[label=P.\arabic*]
\item \label{P1}
The inertia matrices $\fbm{M}_{i}(\fbm{x}_{i})$ are symmetric, positive definite and uniformly bounded by: $\lambda_{i1}\fbm{I}\preceq\fbm{M}_{i}(\fbm{x}_{i})\preceq\lambda_{i2}\fbm{I}$ for any $\fbm{x}_{i}\in\mathbb{R}^{n}$, where $\lambda_{i2}>\lambda_{i1}>0$;
\item \label{P2}
$\dbm{M}_{i}(\fbm{x}_{i})-2\fbm{C}_{i}(\fbm{x}_{i},\dbm{x}_{i})$ are skew-symmetric;
\item \label{P3}
$\fbm{C}_{i}(\fbm{x}_{i},\fbm{y}_{i})$ are linear in $\fbm{y}_{i}$, and there exist $c_{i}>0$ such that $\|\fbm{C}_{i}(\fbm{x}_{i},\fbm{y}_{i})\fbm{z}_{i}\|\leq c_{i}\|\fbm{y}_{i}\|\|\fbm{z}_{i}\|$, $\forall \fbm{x}_{i}, \fbm{y}_{i}, \fbm{z}_{i}\in\mathbb{R}^{n}$.
\end{enumerate}  

Let all slaves in the swarm have the same communication radius $r$. Then, slaves $i$ and $j$ can exchange information at $t\geq 0$ only if their distance is strictly smaller than $r$, i.e., $\|\fbm{x}_{ij}\|=\|\fbm{x}_{i}(t)-\fbm{x}_{j}(t)\|<r$, $\forall i,j\in\{1,\cdots,N\}$. At $t\geq 0$, slaves $i$ and $j$ are adjacent to each other iff they exchange information, i.e., iff their bidirectional communication link $(i,j)$ exists. Because all slaves have communication radius $r$, this paper will maintain certain links $(i,j)$ of the teleoperated swarm by properly constraining $\|\fbm{x}_{ij}(t)\|$ for all time $t\geq 0$.  

%%%%%%%%
\begin{remark}\label{rem1}
\normalfont
This paper starts by preserving connectivity for a teleoperated swarm with no delay, similarly to existing work~\cite{Lee2013TMECH}. Therefore, this section presents the dependency on distance of the inter-slave communications only for a delay-free swarm. \sect{sec: time delay} will reformulate the distance dependency and extend the controller design to maintain the connectivity of a time-delay swarm.
\end{remark}

The information exchanges among the slave robots of a teleoperated swarm can be modelled as an undirected graph $\mathcal{G}(t)=\{\mathcal{V},\mathcal{E}(t)\}$. The vertex set $\mathcal{V}=\{1,\cdots,N\}$ collects all slaves in the swarm. The edge set $\mathcal{E}(t)\subset\{(i,j)\in\mathcal{V}\times\mathcal{V}\}$ includes all inter-slave communication links. By definition, slaves $i$ and $j$ are adjacent at time~$t$, i.e., $(i,j)\in\mathcal{E}(t)$, iff they exchange information at time~$t$. For each slave $i\in\mathcal{V}$, its neighbourhood $\mathcal{N}_{i}(t)=\{j\in\mathcal{V}\ |\ (i,j) \in \mathcal{E}(t)\}$ collects all slaves adjacent to it at time~$t$. In $\mathcal{G}(t)$, a path between vertices $i$ and $j$ is a sequence of distinct vertices $i,a,b,\cdots,j$ such that consecutive vertices are adjacent. By definition, the graph $\mathcal{G}(t)$ is connected iff there exists a path between every two distinct vertices. Further, $\mathcal{G}(t)$ is a tree if every two vertices are connected by exactly one path, as shown in~\fig{fig1a}.

Given a tree $\mathcal{G}(t)$, the associated weighted adjacency matrix $\fbm{A}=[a_{ij}]$ is defined by: $a_{ij}=a_{ji}>0$ if $(i,j)\in\mathcal{E}(t)$, and $a_{ij}=0$ otherwise. The weighted Laplacian matrix $\fbm{L}=[l_{ij}]$ is defined by $l_{ij}=\sum_{k\in\mathcal{N}_{i}(t)}a_{ik}$ if $j=i$, and $l_{ij}=-a_{ij}$ otherwise. For $a_{ij}\in\{0,1\}$, the associated Laplacian becomes an unweighted Laplacian matrix $\bbm{L}$. Let an orientation of $\mathcal{G}(t)$ define an oriented graph $\mathcal{G}^{\ast}(t)$, and label each oriented edge $(i,j)$ as $e_{k}$, $k=1,\cdots,N-1$, with weight $w(e_{k})=a_{ij}=a_{ji}$, as illustrated in~\fig{fig1b}. The associated incidence matrix of $\mathcal{G}^{\ast}(t)$ is $\fbm{D}=[d_{hk}]$ with $d_{hk}=1$ if vertex $h$ is the head of edge $e_{k}$, $d_{hk}=-1$ if $h$ is the tail of $e_{k}$, and $d_{hk}=0$ otherwise. The edge Laplacian of $\mathcal{G}^{\ast}(t)$ is $\fbm{L}_{e}=\tbm{D}\fbm{D}$. 

The following lemmas from~\cite{Egerstedt2010Princeton} will facilitate the proof of~\lem{lem1} in~\sect{sec: set invariance}.
\begin{enumerate}[label=L.\arabic*]
\item \label{L1}
The second smallest eigenvalue $\lambda_{L}$ of the unweighted Laplacian $\bbm{L}$ is positive, $\lambda_{L}>0$.
\item \label{L2}
The set of nonzero eigenvalues of $\fbm{L}_{e}$ is equal to the set of nonzero eigenvalues of the unweighted Laplacian $\bbm{L}$.
\item \label{L3}
The weighted Laplacian $\fbm{L}$ admits the decomposition $\fbm{L}=\fbm{D}\fbm{W}\tbm{D}$, where $\fbm{W}$ is an $(N-1)\times (N-1)$ diagonal matrix with $w(e_{k})$, $k=1,\cdots,N-1$, on the diagonal.
\end{enumerate} 

\begin{figure}[!hbt]
\centering
\subfigure[The tree network $\mathcal{G}(t)$ of a teleoperated slave swarm.]{
	\label{fig1a}
	\includegraphics[width=.45\columnwidth ,height=2.5cm]{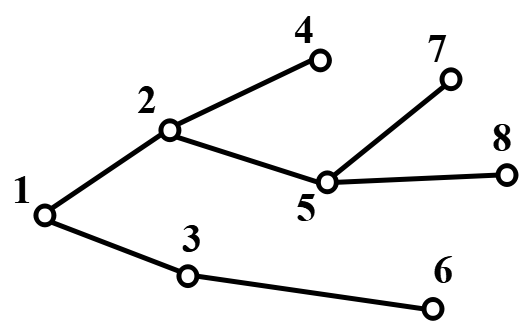}}\quad
\subfigure[An orientation $\mathcal{G}^{\ast}(t)$ of the undirected tree network $\mathcal{G}(t)$.]{
	\label{fig1b}
	\includegraphics[width=.45\columnwidth ,height=2.5cm]{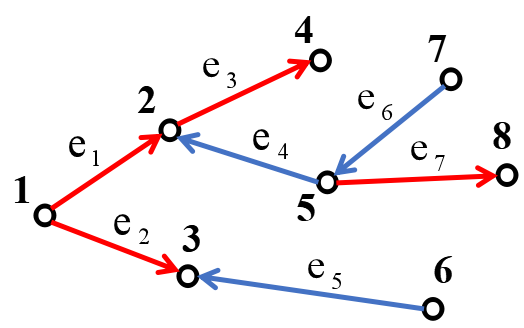}}
\caption{The tree network $\mathcal{G}(t)$ and an orientation $\mathcal{G}^{\ast}(t)$ of a slave swarm with dynamics~\eqref{equ1}.}
\label{fig1}
\end{figure}

Further, the paper adopts the following assumptions on the initial configuration of the system and on the user force.
\begin{assumption}\label{ass1}
The initial connectivity of the slave swarm is a tree $\mathcal{G}(0)$, and each pair of initially adjacent slaves is strictly within their communication distance, i.e., $\forall (i,j)\in\mathcal{E}(0) \Rightarrow \|\fbm{x}_{ij}(0)\|< \hat{r}=r-\epsilon$ for some $\epsilon>0$.
\end{assumption} 
\begin{assumption}\label{ass2}
The user command is bounded by $\|\fbm{f}\|\leq\overline{f}$.
\end{assumption}
\begin{assumption}\label{ass3}
The transmission of information from slave $i$ to slave $j$ is affected by the time-varying delay $T_{ij}(t)$, upper-bounded by $T_{ij}(t)\leq\overline{T}_{ij}$.
\end{assumption}
\begin{assumption}\label{ass4}
Each slave $i$ is initially at rest, and has been at rest longer than the maximum transmission delay from it to its neighbours $j$, i.e., $\dbm{x}_{i}(\tau)=\fbm{0}\ \forall\tau\in[-\overline{T}_{i},0]$ and $\forall i=1,\cdots,N$, where $\overline{T}_{i}=\max\limits_{j\in\mathcal{N}_{i}(0)}\left(\overline{T}_{ij}\right)$.
\end{assumption}

%%%%%%%%
\begin{remark}\label{rem2}
\normalfont 
\ass{ass1} guarantees the initial tree network connectivity $\mathcal{G}(0)$ for any connected swarm~\cite{Egerstedt2010Princeton}. The distance condition $\|\fbm{x}_{ij}(0)\|<\hat{r}$ on all pairs of initially adjacent slaves $(i,j)\in\mathcal{E}(0)$ accounts for the inertia of the EL swarm network, as in the connectivity preservation of autonomous second-order MRS-s~\cite{Su2010SCL}. \ass{ass2} can be guaranteed regardless of unpredictable user actions, by saturating the connection of the informed slave to the master, as in~\sect{sec: experiments}. In \ass{ass3}, the bound on the time-varying delays determines the stabilizing damping for time-delay swarms, as in~\cite{Nuno2009IJRR}. \ass{ass4} facilitates the connectivity analysis for time-delay swarms in~\theo{theorem2}. 
\end{remark}

Then, the paper addresses the following connectivity-preserving swarm teleoperation problem.
\begin{problem}\label{prob1}
Find distributed control laws to drive the teleoperated slave swarm~\eqref{equ1} with~\asses{ass1}{ass4} such that:
\begin{enumerate}
\item[1.]
The velocities of all, and distance between any two, slaves are bounded in the presence of the user command, i.e., $\{\dbm{x}_{i}, \dbm{x}_{j}, \fbm{x}_{i}-\fbm{x}_{j}\}\in\mathcal{L}_{\infty}$, $\forall i,j=1,\cdots,N$ when $\fbm{f}\neq\fbm{0}$;
\item[2.]
All slaves are synchronized in the absence of user force, i.e., $\{\dbm{x}_{i}, \dbm{x}_{j}, \fbm{x}_{i}-\fbm{x}_{j}\}\to\fbm{0}$, $\forall i,j=1,\cdots,N$ when $\fbm{f}=\fbm{0}$;
\item[3.]
All interaction links of the initial network $\mathcal{G}(0)$ are maintained, i.e., $(i,j)\in\mathcal{E}(0) \Rightarrow(i,j)\in\mathcal{E}(t)$ $\forall t\geq 0$, and the connectivity of the slave swarm $\mathcal{G}(t)$ is preserved.
\end{enumerate}
\end{problem}

%%%%%%%%
\begin{remark}\label{rem3}
\normalfont The first two objectives are similar to those of interactive robotic systems~\cite{Lee2010TRO}. Bounded distances between slaves under an operator-applied force indicate that the user can control the slave swarm by commanding the informed slave. Synchronization in the absence of the user command implies that 
swarm teleoperation is a more general MRS problem than the synchronization of autonomous MRS-s. The third objective indicates that the proposed strategy maintains a teleoperated swarm connected by preserving an initial spanning tree of it. Switching spanning trees for connectivity maintenance of teleoperated swarms are a topic of future work.
\end{remark}

%%%%%%%%
\begin{remark}\label{rem4} 
\normalfont The user command $\fbm{f}$ makes the teleoperated swarm~\eqref{equ1} semi-autonomous, and distinguishes its control from that of autonomous and leader-follower MRS-s in two ways.
\begin{enumerate}
\item Compared to the external inputs of autonomous MRS-s, which are disturbances to be rejected, $\fbm{f}$ drives the slave swarm to the desired configuration. The swarm controller should not reject it. Yet, $\fbm{f}$ may endanger the connectivity of the swarm, unintentionally or maliciously. Therefore, the gist of the proposed strategy is to let $\fbm{f}$ drive the swarm to the extent to which $\fbm{f}$ does not threaten its connectivity.
\item Compared to the external input of leader-follower systems, which generates the leader control signal, $\fbm{f}$ is an unpredictable command because it is determined by the unpredictable intentions of the human operator. In contrast, followers are controlled to track the predictable motion of the leader and leader-follower MRS-s are completely self-controlled.
\end{enumerate}
\end{remark}

The proposed solution to \prob{prob1} needs the following definitions and lemma to investigate the robust synchronization of a teleoperated time-delay swarm in~\sect{sec: iss}.  

A function $\alpha:\mathbb{R}_{\geq 0}\mapsto\mathbb{R}_{\geq 0}$ is: of class $\mathcal{K}$ if it is continuous, strictly increasing and $\alpha(0)=0$; of class $\mathcal{K}_{\infty}$ if it is of class $\mathcal{K}$ and unbounded; of class $\mathcal{L}$ if it decreases to zero as its argument increases to $+\infty$. A function $\beta:\mathbb{R}_{\geq 0}\times\mathbb{R}_{\geq 0}\mapsto\mathbb{R}_{\geq 0}$ is of class $\mathcal{KL}$ if it is of class $\mathcal{K}$ in its first argument and of class $\mathcal{L}$ in its second argument. Let $\mathcal{C}([-T,0];\mathbb{R}^{m})$ denote the set of continuous functions defined on $[-T,0]$ and with values in $\mathbb{R}^{m}$. For any essentially bounded function $\bm{\phi}_{t}\in\mathcal{C}([-T,0];\mathbb{R}^{m})$, let 
\begin{align*}
|\bm{\phi}_{t}|_{T}=\sup\limits_{-T\leq\tau\leq 0}\|\bm{\phi}(t+\tau)\|
\end{align*}
and $|\bm{\phi}_{t}|_{a}$ be a norm of $\bm{\phi}$ such that:
\begin{align*}
\gamma_{a}\|\bm{\phi}(t)\|\leq |\bm{\phi}_{t}|_{a}\leq\overline{\gamma}_{a}|\bm{\phi}_{t}|_{T}
\end{align*}
for some positive $\gamma_{a}$ and $\overline{\gamma}_{a}$.

\begin{definition}\label{iss}
~\cite{Jiang2006SCL} The time-delay nonlinear system:
\begin{align*}
\dfbm{\phi}(t)=&\varphi(\bm{\phi}_{t},\fbm{f}(t))\textrm{,}\quad t\geq 0\ a.e.\textrm{,}\\
\bm{\phi}(\tau)=&\bm{\xi}(\tau)\textrm{,}\quad \tau\in[-T,0]\textrm{,}
\end{align*}
with $\bm{\phi}_{t}:[-T,0]\mapsto\mathbb{R}^{m}$ the standard function $\bm{\phi}_{t}(\tau)=\bm{\phi}(t+\tau)$, and $T$ the maximum involved delay, $\varphi:\mathcal{C}([-T,0];\mathbb{R}^{m})\times\mathbb{R}^{l}\mapsto\mathbb{R}^{m}$, and $\bm{\xi}\in\mathcal{C}([-T,0];\mathbb{R}^{m})$, is ISS with input $\fbm{f}(t)\in\mathbb{R}^{l}$ and state $\bm{\phi}(t)\in\mathbb{R}^{m}$ if there exist functions $\alpha\in\mathcal{K}$ and $\beta\in\mathcal{KL}$ such that:
\begin{align*}
\|\bm{\phi}(t)\|\leq\beta\left(|\bm{\xi}|_{T},t\right)+\alpha\left(\sup\limits_{0\leq\tau\leq t}\|\fbm{f}(\tau)\|\right)\textrm{,}\quad \forall t\geq 0\textrm{.}
\end{align*}
\end{definition}
\begin{enumerate}[label=L.4]
\item \label{L4}
~\cite{Jiang2006SCL} The time-delay nonlinear system in~\defi{iss} is ISS if there exists a Lyapunov-Krasovskii functional $V: \mathcal{C}([-T,0];\mathbb{R}^{m})\mapsto\mathbb{R}_{\geq 0}$, functions $\alpha_{1},\alpha_{2}\in\mathcal{K}_{\infty}$, $\alpha_{3},\alpha_{4}\in\mathcal{K}$ such that:
\begin{enumerate}
\item
$\alpha_{1}(\|\bm{\phi}\|)\leq V\leq\alpha_{2}(|\bm{\phi}_{t}|_{a})$, $\forall\bm{\phi}_{t}\in\mathcal{C}([-T,0];\mathbb{R}^{m})$;
\item
$\dot{V}\leq -\alpha_{3}(|\bm{\phi}_{t}|_{a})$, $\forall\bm{\phi}_{t}\in\mathcal{C}([-T,0];\mathbb{R}^{m})$, $\fbm{f}(t)\in\mathbb{R}^{l}:|\bm{\phi}_{t}|_{a}\geq\alpha_{4}(\|\fbm{f}(t)\|)$.
\end{enumerate}
\end{enumerate}

%%%%%%%%%%%%%%%%%%%%%%%%%%%%%%%%%%%%%
%%%%%%%%%%%%%%%%%%%%%%%%%%%%%%%%%%%%%
\section{Main Results}\label{sec: main results}
%%%%%%%%%%%%%%%%%%%%%%%%%%%%%%%%%%%%%
%%%%%%%%%%%%%%%%%%%%%%%%%%%%%%%%%%%%%

This section first introduces a customized potential function for set invariance analysis and proposes a key lemma on the structural controllability of a tree network. Next, it presents the dynamic strategy for preserving the tree network, and implicitly the connectivity, of a teleoperated swarm without and with time-varying delays. Lastly, it proves that the proposed dynamic control renders the time-delay slave swarm ISS and thus guarantees its robust synchronization.

%%%%%%%%%%%%%%%%%%%%%%%%%%%%%%%%%%%%%
\subsection{Set Invariance by Energy Bounding}\label{sec: set invariance}
%%%%%%%%%%%%%%%%%%%%%%%%%%%%%%%%%%%%%

By~\ass{ass1} and the third objective in~\prob{prob1}, the connectivity of the slave swarm~\eqref{equ1} can be maintained by rendering invariant the edge set of $\mathcal{G}(0)$, $\mathcal{E}(t)=\mathcal{E}(0)\ \forall t\geq 0$. Because the communication links between slaves $(i,j)\in\mathcal{E}(0)$ are constrained by distance, the following function verifies the distance constraints:
\begin{equation}\label{equ2}
\psi(\|\fbm{x}_{ij}\|)=\frac{P\|\fbm{x}_{ij}\|^{2}}{r^{2}-\|\fbm{x}_{ij}\|^{2}+Q}\textrm{,}
\end{equation}
where $P$ and $Q$ are positive constants to be designed. The function $\psi(\|\fbm{x}_{ij}\|)$ is continuous, positive definite and strictly increasing with respect to $\|\fbm{x}_{ij}\|\in[0,r]$~\cite{Su2010SCL}. Therefore, it can quantify the energy stored in the slave network by:
\begin{equation}\label{equ3}
V_{p}=\frac{1}{2}\sum^{N}_{i=1}\sum_{j\in\mathcal{N}_{i}(0)}\psi(\|\fbm{x}_{ij}\|)\textrm{.}
\end{equation}

The following proposition offers a path to upper-bounding the swarm energy as quantified by functions~\eqref{equ2} and~\eqref{equ3}:
\begin{proposition}\label{prop1} 
For any network with \ass{ass1}, its energy $V_{p}$ in~\eqref{equ3} satisfies:
\begin{align*}
V_{p}(0)+\Delta<\frac{P\overline{r}^{2}}{r^{2}-\overline{r}^{2}+Q}=\psi_{\max}
\end{align*}
for any $\Delta>0$ and any positive $P$, $Q$ and $\overline{r}=r-\kappa\epsilon$ with $\kappa\in[0,1)$ satisfying:
\begin{equation}\label{equ4}
\begin{aligned}
\omega_{1}=&\left(r^{2}-\hat{r}^{2}\right)\overline{r}^{2}-(N-1)\left(r^{2}-\overline{r}^{2}\right)\hat{r}^{2}>0\textrm{,}\\
\omega_{2}=&\omega_{1}+\Big(\overline{r}^{2}-(N-1)\hat{r}^{2}\Big)Q>0\textrm{,}\\
\omega_{3}=&P\omega_{2}-\left(r^{2}-\overline{r}^{2}+Q\right)\left(r^{2}-\hat{r}^{2}+Q\right)\Delta>0\textrm{.}
\end{aligned}
\end{equation}
\end{proposition}
\begin{proof}
\ass{ass1} and the property that $\psi(\|\fbm{x}_{ij}\|)$ in~\eqref{equ2} is strictly increasing on $\|\fbm{x}_{ij}\|\in [0,r]$ imply that:
\begin{align*}
V_{p}(0)<\frac{1}{2}\sum^{N}_{i=1}\sum_{j\in\mathcal{N}_{i}(0)}\frac{P\hat{r}^{2}}{r^{2}-\hat{r}^{2}+Q}=\frac{(N-1)P\hat{r}^{2}}{r^{2}-\hat{r}^{2}+Q}\textrm{,}
\end{align*} 
where $N-1$ is the number of communication links in $\mathcal{G}(0)$. A sufficiently small $\kappa\in[0,1)$ makes $\omega_{1}>0$. Lastly, after selecting small enough $Q>0$ to satisfy $\omega_{2}>0$, a sufficiently large $P>0$ makes $\omega_{3}>0$ and guarantees that:
\begin{align*}
&\psi_{\max}-V_{p}(0)> \frac{P\overline{r}^{2}}{r^{2}-\overline{r}^{2}+Q}-\frac{(N-1)P\hat{r}^{2}}{r^{2}-\hat{r}^{2}+Q}\\
=&\frac{\Big(\left(r^{2}-\hat{r}^{2}+Q\right)\overline{r}^{2}-(N-1)\left(r^{2}-\overline{r}^{2}+Q\right)\hat{r}^{2}\Big)P}{\left(r^{2}-\overline{r}^{2}+Q\right)\left(r^{2}-\hat{r}^{2}+Q\right)}\\
=&\frac{P\omega_{2}}{\left(r^{2}-\overline{r}^{2}+Q\right)\left(r^{2}-\hat{r}^{2}+Q\right)}>\Delta\textrm{.}
\end{align*}
\end{proof}

\prop{prop1} contributes to examining the distance constraint on every link $(i,j)\in\mathcal{E}(0)$ by \prop{prop2} below.
\begin{proposition}\label{prop2}
If $P$, $Q$ and $\overline{r}=r-\kappa\epsilon$ satisfy condition~\eqref{equ4} and: 
\begin{equation}\label{equ5}
V_{p}(\tau)\leq V_{p}(0)+\Delta\textrm{,}\quad \forall\tau\in[0,t]\textrm{,}
\end{equation}
then $\|\fbm{x}_{ij}(t)\|<\overline{r}$ for every $(i,j)\in\mathcal{E}(0)$.
\end{proposition}
\begin{proof}
By~\prop{prop1}, \eq{equ5} implies that:
\begin{align*}
V_{p}(\tau)<\psi_{\max}\textrm{,}\quad \forall \tau\in[0,t]\textrm{.}
\end{align*}
\ass{ass1} and the strict monotonicity of~\eqref{equ2} then lead to:
\begin{align*}
0\leq\psi(\|\fbm{x}_{ij}(0)\|)<\psi_{\max}\textrm{,}\quad \forall (i,j)\in\mathcal{E}(0)\textrm{.}
\end{align*}
Suppose that, at time $t$, $\|\fbm{x}_{ij}(t)\|=\overline{r}$ is the longest distance among all pairs of initially neighbouring slaves. It then follows that $\psi(\|\fbm{x}_{ij}(t)\|)=\psi_{\max}$ and $0\leq\psi(\|\fbm{x}_{lm}(t)\|)\leq \psi_{\max}$ for any other $(l,m)\in\mathcal{E}(0)$, and further that:
\begin{align*}
V_{p}(t)=\psi(\|\fbm{x}_{ij}(t)\|)+\sum_{(l,m)\in\overline{\mathcal{E}}(0)}\psi(\|\fbm{x}_{lm}(t)\|)\geq \psi_{\max}\textrm{,}
\end{align*}
for $\overline{\mathcal{E}}(0)=\mathcal{E}(0)-\{(i,j)\}$, which contradicts~\eqref{equ5} and, thus, leads to the conclusion that $\|\fbm{x}_{ij}(t)\|<\overline{r}\ \forall (i,j)\in\mathcal{E}(0)$.
\end{proof}

%%%%%%%%
\begin{remark}\label{rem5}
\normalfont
Propositions~\ref{prop1} and~\ref{prop2} use the customized potentials~\eqref{equ2} and~\eqref{equ3}  to verify the invariance~\cite{Blanchini1999Auto,Ames2017TAC} of the initial edge set $\mathcal{E}(0)$ of a teleoperated swarm. Because the potentials~\eqref{equ2} and~\eqref{equ3} depend on the inter-slave distances, this paper constrains the distance $\|\fbm{x}_{ij}\|$ between each pair of initially adjacent slaves $(i,j)\in\mathcal{E}(0)$ so as to render the edge set $\mathcal{E}(0)$ invariant for both delay-free and time-delay teleoperated swarms. The proposed control policy operates similarly to the strategy for connectivity maintenance of autonomous double-integrator multi-agent systems~\cite{Su2010SCL}. However, it preserves the connectivity of teleoperated swarms driven by unpredictable user forces, similarly to virtual point-based control~\cite{Lee2013TMECH}. Uniquely, the constructive design in this paper preserves the connectivity of teleoperated swarms even with time-varying delays in the inter-slave transmissions. 
\end{remark}

To bound the potential energy $V_{p}$ using only local slave information, define a surface $\fbm{s}_{i}$:
\begin{equation}\label{equ6}
\fbm{s}_{i}=\dbm{x}_{i}+\sigma\bm{\theta}_{i}\textrm{,}
\end{equation}
for every slave $i=1,\cdots,N$, with $\sigma>0$ and:
\begin{equation}\label{equ7}
\bm{\theta}_{i}=\sum_{j\in\mathcal{N}_{i}(0)}\nabla_{i}\psi(\|\fbm{x}_{ij}\|)
\textrm{,}
\end{equation}
where the gradient of $\psi(\|\fbm{x}_{ij}\|)$ with respect to $\fbm{x}_{i}$ is given by:
\begin{equation}\label{equ8}
\nabla_{i}\psi(\|\fbm{x}_{ij}\|)=\frac{2P\cdot\left(r^{2}+Q\right)}{\left(r^{2}-\|\fbm{x}_{ij}\|^{2}+Q\right)^{2}}\left(\fbm{x}_{i}-\fbm{x}_{j}\right)\textrm{.}
\end{equation}
Then, the dynamics of the teleoperated swarm~\eqref{equ1} become:
\begin{equation}\label{equ9}
\begin{aligned}
\fbm{M}_{1}(\fbm{x}_{1})\dbm{s}_{1}+\fbm{C}_{1}(\fbm{x}_{1},\dbm{x}_{1})\fbm{s}_{1}=&\sigma\fbm{\Delta}_{1}+\fbm{u}_{1}+\fbm{f}\textrm{,}\\
\fbm{M}_{s}(\fbm{x}_{s})\dbm{s}_{s}+\fbm{C}_{s}(\fbm{x}_{s},\dbm{x}_{s})\fbm{s}_{s}=&\sigma\fbm{\Delta}_{s}+\fbm{u}_{s}\textrm{,}
\end{aligned}
\end{equation}
where $s=2,\cdots,N$, and $\bm{\Delta}_{i}$ are state-dependent mismatches:
\begin{equation}\label{equ10}
\fbm{\Delta}_{i}=\fbm{M}_{i}(\fbm{x}_{i})\dfbm{\theta}_{i}+\fbm{C}_{i}(\fbm{x}_{i},\dbm{x}_{i})\bm{\theta}_{i}\ \forall i=1,\cdots,N\textrm{.}
\end{equation}

The following lemma will play a key role in the proof of connectivity maintenance in the remainder of the paper.
%%%%%%%%
\begin{lemma}\label{lem1}
For the teleoperated swarm~\eqref{equ1} with the tree network $\mathcal{G}(0)$ and \asses{ass1}{ass2}, the following inequality holds:
\begin{equation}\label{equ11}
\begin{aligned}
\sum^{N}_{i=1}\tfbm{\theta}_{i}\bm{\theta}_{i}\geq \frac{4\lambda_{L}P}{r^{2}+Q} V_{p}\textrm{.}
\end{aligned}
\end{equation}
\end{lemma}
\begin{proof}
The weighted adjacency matrix $\fbm{A}=[a_{ij}]$ associated with the tree $\mathcal{G}(0)$ is: 
\begin{align*}
a_{ij}=\begin{cases}
\frac{2P\cdot(r^{2}+Q)}{\left(r^{2}-\|\fbm{x}_{ij}\|^{2}+Q\right)^{2}}\quad &\text{if } j\in\mathcal{N}_{i}(0)\\
0\quad &\text{otherwise}
\end{cases}
\textrm{,}
\end{align*}
and the corresponding weighted Laplacian is $\fbm{L}=[l_{ij}]$ with:
\begin{align*}
l_{ij}=\begin{cases}
-a_{ij}\quad &\text{if } j\neq i\\
\sum_{k\in\mathcal{N}_{i}(0)}a_{ik}\quad &\text{else if } j=i 
\end{cases}
\textrm{.} 
\end{align*}

Letting $\fbm{l}_{i}$ be the $i$-th row of $\fbm{L}$, it follows that:
\begin{align*}
\sum_{j\in\mathcal{N}_{i}(0)}\nabla_{i}\psi(\|\fbm{x}_{ij}\|)=\left(\fbm{l}_{i}\otimes\fbm{I}_{n}\right)\fbm{x}\textrm{,}
\end{align*}
for $\fbm{x}=\left(\tbm{x}_{1}\ \cdots\ \tbm{x}_{N}\right)^\mathsf{T}$. After using the definition~\eqref{equ7} of $\bm{\theta}_{i}$, the left side of~\eqref{equ11} becomes:
\begin{align*}
\sum^{N}_{i=1}\tfbm{\theta}_{i}\bm{\theta}_{i}&=\sum^{N}_{i=1}\tbm{x}\left(\fbm{l}_{i}\otimes\fbm{I}_{n}\right)^\mathsf{T}\left(\fbm{l}_{i}\otimes\fbm{I}_{n}\right)\fbm{x}\\
&=\tbm{x}\left(\left(\sum^{N}_{i=1}\tbm{l}_{i}\fbm{l}_{i}\right)\otimes\fbm{I}_{n}\right)\fbm{x}=\tbm{x}\left(\tbm{L}\fbm{L}\otimes\fbm{I}_{n}\right)\fbm{x}\textrm{.}
\end{align*}
By~Lemma~\ref{L3} and the definition of $\fbm{W}$ in~\cite{Egerstedt2010Princeton}, the left-hand side of~\eqref{equ11} can be re-organized further as:
\begin{equation}\label{equ12}
\begin{aligned}
&\sum^{N}_{i=1}\tfbm{\theta}_{i}\bm{\theta}_{i}=\tbm{x}\left(\fbm{D}\fbm{W}\tbm{D}\fbm{D}\fbm{W}\tbm{D}\otimes\fbm{I}_{n}\right)\fbm{x}\\
=&\left(\left(\fbm{W}\tbm{D}\otimes\fbm{I}_{n}\right)\fbm{x}\right)^\mathsf{T}\left(\tbm{D}\fbm{D}\otimes\fbm{I}_{n}\right)\left(\left(\fbm{W}\tbm{D}\otimes\fbm{I}_{n}\right)\fbm{x}\right)\\
=&\tbbm{x}\left(\fbm{L}_{e}\otimes\fbm{I}_{n}\right)\bbm{x}\textrm{,}
\end{aligned}
\end{equation}
where $\bbm{x}=(\tbbm{x}_{1}\ \cdots\ \tbbm{x}_{N-1})^\mathsf{T}=(\fbm{W}\tbm{D}\otimes\fbm{I}_{n})\fbm{x}$ with $\bbm{x}_{k}=\nabla_{i}\psi(\|\fbm{x}_{ij}\|)$ for $e_{k}=(i,j)$, $k=1,\cdots,N-1$, stacks the weighted displacements between all pairs of adjacent slaves $(i,j)\in\mathcal{E}(0)$~\cite{Egerstedt2010Princeton}. 

Lemma~\ref{L1} in~\cite{Egerstedt2010Princeton} and \ass{ass1} imply that the second smallest eigenvalue $\lambda_{L}$ of the unweighted laplacian $\overline{\fbm{L}}$ is positive, i.e., $\lambda_{L}>0$. Further, $\fbm{L}_{e}$ is $(N-1)\times (N-1)$, because $\mathcal{G}(0)$ is a tree, and has smallest eigenvalue $\lambda_{L}$ by~Lemma~\ref{L2} in~\cite{Egerstedt2010Princeton}. Hence, the left side of~\eqref{equ13} can be lower-bounded by:
\begin{align*}
&\sum^{N}_{i=1}\tfbm{\theta}_{i}\bm{\theta}_{i}\geq  \lambda_{L}\tbbm{x}\bbm{x}=\lambda_{L}\sum_{(i,j)\in\mathcal{E}(0)}\Big\|\nabla_{i}\psi(\|\fbm{x}_{ij}\|)\Big\|^{2}\\
=&\sum_{(i,j)\in\mathcal{E}(0)}\frac{4\lambda_{L}P\cdot(r^{2}+Q)^{2}}{\left(r^{2}-\|\fbm{x}_{ij}\|^{2}+Q\right)^{3}}\psi(\|\fbm{x}_{ij}\|)\geq \frac{4\lambda_{L}P}{r^{2}+Q} V_{p}\textrm{.}
\end{align*} 
\end{proof}

%%%%%%%%%%%%%%%%%%%%%%%%%%%%%%%%%%%%%
\subsection{Delay-Free Slave Swarm}\label{sec: delay free}
%%%%%%%%%%%%%%%%%%%%%%%%%%%%%%%%%%%%%

Because, in a delay-free swarm, each slave~$i$ receives the information sent by its neighbours $j\in\mathcal{N}_{i}(0)$ if $\|\fbm{x}_{ij}(t)\|<r\ \forall t\geq 0$, $\psi(\|\fbm{x}_{ij}\|)$ in~\eqref{equ2} and $V_{p}$ in~\eqref{equ3} serve to design a connectivity-preserving swarm controller in this section.

Assume that all links $(i,j)\in\mathcal{E}(0)$ have been maintained during the time interval $[0,t)$, i.e., $\|\fbm{x}_{ij}(\tau)\|<r$ $\forall \tau\in[0,t)$ and $\forall (i,j)\in\mathcal{E}(0)$. This assumption implies that the position $\fbm{x}_{j}(t)$ of slave $j\in\mathcal{N}_{i}(0)$ can be employed in the control of slave $i$ at time $t$ to guarantee that $\|\fbm{x}_{ij}(t)\|<r$ $\forall (i,j)\in\mathcal{E}(0)$ by~\prop{prop2}. Then, connectivity maintenance can be proven by induction on time~\cite{Su2010SCL}. The control proposed to render the edge set $\mathcal{E}(0)$ positively invariant is:
\begin{equation}\label{equ13}
\fbm{u}_{i}=-K_{i}(t)\fbm{s}_{i}-D\dbm{x}_{i}-B\bm{\theta}_{i}\textrm{,}
\end{equation}  
where $K_{i}(t)$, $D$ and $B$ are positive gains to be determined. 

%%%%%%%%
\begin{remark}\label{rem6}
\normalfont
After using the definition~\eqref{equ6} of $\fbm{s}_{i}$, the control $\fbm{u}_{i}$ becomes:
\begin{align*}
\fbm{u}_{i}=-\Big(\sigma K_{i}(t)+B\Big)\bm{\theta}_{i}-\Big(K_{i}(t)+D\Big)\dbm{x}_{i}\textrm{,}
\end{align*} 
where the first term is the coupling force and the second term is the injected damping force. The gain $K_{i}(t)$ is updated dynamically, according to the distances $\|\fbm{x}_{ij}\|$ between slave $i$ and its neighbours $j\in\mathcal{N}_{i}(0)$.
\end{remark}

The following Lyapunov candidate is proposed to investigate connectivity preservation:
\begin{equation}\label{equ14}
\begin{aligned}
V_{1}=\frac{1}{2(B+\sigma D)}\sum^{N}_{i=1}\tbm{s}_{i}\fbm{M}_{i}(\fbm{x}_{i})\fbm{s}_{i}+V_{p}\textrm{,}
\end{aligned}
\end{equation}
with $V_{p}$ defined in~\eqref{equ3}. After using Properties~\ref{P1} and~\ref{P2} of~\eqref{equ1} and~\ass{ass1}, the derivative of $V_{1}$ along the trajectories of~\eqref{equ9} in closed-loop with~\eqref{equ13} is:
\begin{align*}
\dot{V}_{1}=&\frac{1}{2(B+\sigma D)}\sum^{N}_{i=1}\left(\tbm{s}_{i}\dbm{M}_{i}(\fbm{x}_{i})\fbm{s}_{i}+2\tbm{s}_{i}\fbm{M}_{i}(\fbm{x}_{i})\dbm{s}_{i}\right)\\
&+\frac{1}{2}\sum^{N}_{i=1}\sum_{j\in\mathcal{N}_{i}(0)}\left(\dtbm{x}_{i}\nabla_{i}\psi(\|\fbm{x}_{ij}\|)+\dtbm{x}_{j}\nabla_{j}\psi(\|\fbm{x}_{ij}\|)\right)\\
=&\sum^{N}_{i=1}\frac{\sigma\tbm{s}_{i}\fbm{\Delta}_{i}-K_{i}(t)\tbm{s}_{i}\fbm{s}_{i}}{B+\sigma D}+\frac{\tbm{s}_{1}\fbm{f}}{B+\sigma D}\\
&-\sum^{N}_{i=1}\frac{\tbm{s}_{i}\left(D\dbm{x}_{i}+B\bm{\theta}_{i}\right)}{B+\sigma D}+\sum^{N}_{i=1}\sum_{j\in\mathcal{N}_{i}(0)}\dtbm{x}_{i}\nabla_{i}\psi(\|\fbm{x}_{ij}\|)\textrm{.}
\end{align*}
The assumption that $\mathcal{E}(\tau)=\mathcal{E}(0)$ for all $\tau\in [0,t)$ implies that $\bm{\theta}_{i}$ in~\eqref{equ7} can be used in the control $\fbm{u}_{i}$ at time~$t$ and, together with the definition~\eqref{equ6} of $\fbm{s}_{i}$, leads to:
\begin{align*}
&\tbm{s}_{i}(D\dbm{x}_{i}+B\bm{\theta}_{i})\\
=&D\dtbm{x}_{i}\dbm{x}_{i}+\sigma D\dtbm{x}_{i}\bm{\theta}_{i}+B\dtbm{x}_{i}\bm{\theta}_{i}+\sigma B\tfbm{\theta}_{i}\bm{\theta}_{i}\\
=&D\dtbm{x}_{i}\dbm{x}_{i}+\sigma B\tfbm{\theta}_{i}\bm{\theta}_{i}+(B+\sigma D)\sum_{j\in\mathcal{N}_{i}(0)}\dtbm{x}_{i}\nabla_{i}\psi(\|\fbm{x}_{ij}\|)\textrm{.}
\end{align*}
Thus, $\dot{V}_{1}$ becomes:
\begin{equation}\label{equ15}
\begin{aligned}
\dot{V}_{1}=&\sum^{N}_{i=1}\frac{\sigma\tbm{s}_{i}\bm{\Delta}_{i}-\sigma B\tfbm{\theta}_{i}\bm{\theta}_{i}}{B+\sigma D}+\frac{\tbm{s}_{1}\fbm{f}}{B+\sigma D}\\
&-\sum^{N}_{i=1}\frac{1}{B+\sigma D}\Big(K_{i}(t)\tbm{s}_{i}\fbm{s}_{i}+D\dtbm{x}_{i}\dbm{x}_{i}\Big)\textrm{.}
\end{aligned}
\end{equation}
The definition of $\bm{\theta}_{i}$ in~\eqref{equ7} further implies that:
\begin{equation}\label{equ16}
\begin{aligned}
\dfbm{\theta}_{i}=&\sum_{j\in\mathcal{N}_{i}(0)}\frac{8P\cdot(r^{2}+Q)\tbm{x}_{ij}\dbm{x}_{ij}\fbm{x}_{ij}}{\left(r^{2}-\|\fbm{x}_{ij}\|^{2}+Q\right)^{3}}\\
&+\sum_{j\in\mathcal{N}_{i}(0)}\frac{2P\cdot(r^{2}+Q)(\dbm{x}_{i}-\dbm{x}_{j})}{\left(r^{2}-\|\fbm{x}_{ij}\|^{2}+Q\right)^{2}}\textrm{,}
\end{aligned}
\end{equation}
and algebraic manipulations using Properties~\ref{P1} and~\ref{P3} of~\eqref{equ1} lead to:
\begin{equation}\label{equ17}
\begin{aligned}
&\tbm{s}_{i}\fbm{M}_{i}(\fbm{x}_{i})\dfbm{\theta}_{i}\leq \sum_{j\in\mathcal{N}_{i}(0)}2\left(\eta_{i}+\gamma_{i}\right)\left(\dtbm{x}_{i}\dbm{x}_{i}+\dtbm{x}_{j}\dbm{x}_{j}\right)\\
&+\sum_{j\in\mathcal{N}_{i}(0)}\Bigg(\frac{16\lambda^{2}_{i2}P^{2}\cdot(r^{2}+Q)^{2}\|\fbm{x}_{ij}\|^{4}}{\eta_{i}\cdot\left(r^{2}-\|\fbm{x}_{ij}\|^{2}+Q\right)^{6}}\tbm{s}_{i}\fbm{s}_{i}\\
&\quad \quad \quad \quad \quad +\frac{\lambda^{2}_{i2}P^{2}\cdot\left(r^{2}+Q\right)^{2}}{\gamma_{i}\cdot\left(r^{2}-\|\fbm{x}_{ij}\|^{2}+Q\right)^{4}}\tbm{s}_{i}\fbm{s}_{i}\Bigg)
\end{aligned}
\end{equation}
with $\eta_{i}>0$ and $\gamma_{i}>0$, and to:
\begin{equation}\label{equ18}
\begin{aligned}
&\tbm{s}_{i}\fbm{C}_{i}(\fbm{x}_{i},\dbm{x}_{i})\bm{\theta}_{i}\\
\leq &\sum_{j\in\mathcal{N}_{i}(0)}\Bigg(\frac{c^{2}_{i}P^{2}\cdot\left(r^{2}+Q\right)^{2}\|\fbm{x}_{ij}\|^{2}}{2\zeta_{i}\cdot\left(r^{2}-\|\fbm{x}_{ij}\|^{2}+Q\right)^{4}}\tbm{s}_{i}\fbm{s}_{i}+2\zeta_{i}\dtbm{x}_{i}\dbm{x}_{i}\Bigg)
\end{aligned}
\end{equation}
with $\zeta_{i}>0$. Hence, the impact of the state-dependent mismatch $\fbm{\Delta}_{i}$ in~\eqref{equ10} can be upper-bounded by:
\begin{equation}\label{equ19}
\begin{aligned}
\tbm{s}_{i}\fbm{\Delta}_{i}\leq&\sum_{j\in\mathcal{N}_{i}(0)}\Big(\Lambda_{ij1}(t)\tbm{s}_{i}\fbm{s}_{i}+2(\eta_{i}+\gamma_{i})\dtbm{x}_{j}\dbm{x}_{j}\\
&\quad \quad \quad \quad +2(\eta_{i}+\gamma_{i}+\zeta_{i})\dtbm{x}_{i}\dbm{x}_{i}\Big)\textrm{,}
\end{aligned}
\end{equation}
where: 
\begin{align*}
&\Lambda_{ij1}(t)=\frac{16\lambda^{2}_{i2}P^{2}\cdot(r^{2}+Q)^{2}\|\fbm{x}_{ij}\|^{4}}{\eta_{i}\cdot\left(r^{2}-\|\fbm{x}_{ij}\|^{2}+Q\right)^{6}}\\
&+\frac{\lambda^{2}_{i2}P^{2}\cdot\left(r^{2}+Q\right)^{2}}{\gamma_{i}\cdot\left(r^{2}-\|\fbm{x}_{ij}\|^{2}+Q\right)^{4}}+\frac{c^{2}_{i}P^{2}\cdot\left(r^{2}+Q\right)^{2}\|\fbm{x}_{ij}\|^{2}}{2\zeta_{i}\cdot\left(r^{2}-\|\fbm{x}_{ij}\|^{2}+Q\right)^{4}}\textrm{.}
\end{align*}
Because the user-injected energy is upper-bounded by:
\begin{equation}\label{equ20}
\tbm{s}_{1}\fbm{f}\leq \|\fbm{f}\|^{2}+\frac{1}{4}\tbm{s}_{1}\fbm{s}_{1}\textrm{,}
\end{equation} 
the following upper bound on $\dot{V}_{1}$ follows after substitution from~\eqref{equ19} and~\eqref{equ20} into~\eqref{equ15}:
\begin{equation}\label{equ21}
\begin{aligned}
\dot{V}_{1}\leq&\sum^{N}_{i=1}\frac{\varrho_{i}\|\fbm{f}\|^{2}-\widehat{K}_{i}(t)\tbm{s}_{i}\fbm{s}_{i}-\widehat{D}_{i}\dtbm{x}_{i}\dbm{x}_{i}-\sigma B\tfbm{\theta}_{i}\bm{\theta}_{i}}{B+\sigma D}\textrm{.}
\end{aligned}
\end{equation}  
In~\eqref{equ21}, $\varrho_{i}=1$ if $i=1$ and $\varrho_{i}=0$ otherwise, and:
\begin{equation}\label{equ22}
\begin{aligned}
&\widehat{K}_{i}(t)= K_{i}(t)-\sum_{j\in\mathcal{N}_{i}(0)}\sigma\Lambda_{ij1}(t)-\frac{\varrho_{i}}{4}\textrm{,}\\
&\widehat{D}_{i}=D-\sum_{j\in\mathcal{N}_{i}(0)}2\sigma\cdot(\eta_{i}+\gamma_{i}+\zeta_{i}+\eta_{j}+\gamma_{j})\textrm{.}
\end{aligned}
\end{equation}

The transformation of the swarm dynamics~\eqref{equ1} into~\eqref{equ9} by the surface~\eqref{equ6} is the first key step of the design proposed in this paper. This model order reduction technique has already served to address system uncertainties in adaptive control~\cite{Slotine2009TRO,Nuno2011TAC}, and to develop feedback $r$-passivity~\cite{Lee2011ICRA} to ensure passivity of a master-informed slave connection with kinematic dissimilarity. Uniquely, this paper uses the technique to preserve the connectivity of a teleoperated swarm without compensation of the system dynamics. The reduced-order swarm dynamics~\eqref{equ9} facilitate the design of the controller~\eqref{equ13} and the analysis of connectivity maintenance as follows. Connectivity preservation for a delay-free swarm with tree network $\mathcal{G}(0)$ is equivalent to rendering the set $\mathcal{E}(0)$ invariant. The invariance of $\mathcal{E}(0)$ is examined by the potential function $V_{p}$ as in~\eqref{equ5} in \prop{prop2}. An upper-bound on $V_{p}$ is provided by the Lyapunov candidate $V_{1}$ in~\eqref{equ14}, which also accounts for the inertia of the EL slaves through its term dependent on the kinetic energy of the dynamics~\eqref{equ9}. The dynamics~\eqref{equ9} in closed-loop with the designed control~\eqref{equ13} upper-bound $\dot{V}_{1}$ as in~\eqref{equ21}. Lastly, the bound~\eqref{equ21} on $\dot{V}_{1}$ is critical to quantifying the upper bound of $V_{1}$ and thus of $V_{p}$ in~\theo{theorem1}. 

Given~\lem{lem1}, the invariance of the edge set $\mathcal{E}(0)$ and, with it, coonectivity preservation for the delay-free teleoperated swarm~\eqref{equ1} is validated by the following theorem.
%%%%%%%%
\begin{theorem}\label{theorem1}
The control~\eqref{equ13} maintains the connectivity of the teleoperated swarm~\eqref{equ1} with \asses{ass1}{ass2} if designed by:
\begin{enumerate}
\item[1. ]
choose $\rho$, $\sigma$, $\eta_{i}$, $\gamma_{i}$, $\zeta_{i}$ and $B$ heuristically;
\item[2. ]
set $D$ to make $\widehat{D}_{i}\geq 0$ in~\eqref{equ22};
\item[3. ]
let $\kappa=0$ and select $Q$ to make $\omega_{2}>0$ in~\eqref{equ4};
\item[4. ]
choose $P$ sufficiently large such that:
\begin{equation}\label{equ23}
P\geq \frac{\rho\cdot(r^{2}+Q)(B+\sigma D)}{4\sigma\lambda_{L}B}\textrm{,}
\end{equation}
and such that $\omega_{3}>0$ in~\eqref{equ4} with:
\begin{align*}
\Delta=\frac{1}{2(B+\sigma D)}\sum^{N}_{i=1}\lambda_{i2}\|\fbm{s}_{i}(0)\|^{2}+\frac{\overline{f}^{2}}{\rho\cdot(B+\sigma D)}\textrm{;}
\end{align*}
\item[5. ]
update $K_{i}(t)$ according to~\eqref{equ22} to guanratee:
\begin{equation}\label{equ24}
\widehat{K}_{i}(t)\geq \frac{1}{2}\rho\lambda_{i2}\textrm{.}
\end{equation}
\end{enumerate}
\end{theorem}
\begin{proof}
After substitution from~\eqref{equ11} and the use of \ass{ass2}, $\widehat{D}_{i}\geq 0$, and \eqref{equ23}-\eqref{equ24}, \eqref{equ21} yields:
\begin{equation}\label{equ25}
\begin{aligned}
\dot{V}_{1}\leq &-\sum^{N}_{i=1}\frac{\widehat{K}_{i}(t)}{B+\sigma D}\tbm{s}_{i}\fbm{s}_{i}-\sum^{N}_{i=1}\frac{\widehat{D}_{i}}{B+\sigma D}\dtbm{x}_{i}\dbm{x}_{i}\\
&-\frac{4\sigma\lambda_{L}PB}{(r^{2}+Q)(B+\sigma D)} V_{p}+\frac{\|\fbm{f}\|^{2}}{B+\sigma D}\\
\leq &-\frac{\rho}{2(B+\sigma D)}\sum^{N}_{i=1}\lambda_{i2}\tbm{s}_{i}\fbm{s}_{i}-\rho V_{p}+\frac{\|\fbm{f}\|^{2}}{B+\sigma D}\\
\leq &-\rho V_{1}+\frac{\overline{f}^{2}}{B+\sigma D}\textrm{,}
\end{aligned}
\end{equation}
Time integration of $\dot{V}_{1}$ from $0$ to $t\geq 0$ then leads to:
\begin{equation}\label{equ26}
\begin{aligned}
V_{1}(t)\leq e^{-\rho t} V_{1}(0)+\frac{\overline{f}^{2}}{\rho\cdot(B+\sigma D)}\textrm{,}
\end{aligned}
\end{equation}
which, together with the definition of $\Delta$, yields:
\begin{align*}
V_{p}(t)\leq V_{1}(t)\leq e^{-\rho t} V_{p}(0)+\Delta\ \forall t\geq 0\textrm{.}
\end{align*}
If $Q$ and $P$ satisfy~\eqref{equ4}, \prop{prop2} implies that $\|\fbm{x}_{ij}(t)\|<r\ \forall (i,j)\in\mathcal{E}(0)$, i.e., the swarm connectivity is preserved.
\end{proof}

%%%%%%%%
\begin{remark}\label{rem7}
\normalfont
The connectivity maintenance proof employs set invariance analysis~\cite{Blanchini1999Auto,Ames2017TAC}. Namely, based on the assumption that $\mathcal{E}(\tau)=\mathcal{E}(0)\ \forall \tau\in [0,t)$, the control $\fbm{u}_{i}$ of slave $i$ at time $t$ can employ the position of slave $j\in\mathcal{N}_{i}(0)$. If the controls $\fbm{u}_{i}\ i=1,\cdots,N$ guarantee that $V_{p}(t)<\psi_{\max}$ as in~\theo{theorem1}, then $\mathcal{E}(t)=\mathcal{E}(0)$. Thus, the controls $\fbm{u}_{i}\ i=1,\cdots,N$ render $\mathcal{E}(0)$ invariant by induction on time. Induction on time has also served to prove connectivity preservation for networks of second-order integrators~\cite{Su2010SCL} and to ensure the safety of locally Lipschitz nonlinear systems~\cite{Ames2017TAC}. 
\end{remark}

The dynamic strategy~\eqref{equ13} tackles connectivity-preserving swarm teleoperation by transforming the dynamics~\eqref{equ1} into the first-order passive form~\eqref{equ9}. This model order reduction is inspired by classical feedback passivation~\cite{Sepulchre2012Springer}. This paper, however, uses it to upper-bound, by~\eqref{equ19}, the impact of state-dependent mismatches $\fbm{\Delta}_{i}$ on swarm connectivity without compensation of system dynamics. By~\eqref{equ22} and~\eqref{equ24}, the dynamic updating of $K_{i}(t)$  in~\eqref{equ13} is sufficient to suppress the impact of $\fbm{\Delta}_{i}$ on the connectivity of the teleoperated swarm. 

\lem{lem1} reveals the structural controllability of a tree $\mathcal{G}(0)$ of an EL network. It shows that the structure of a connected EL network with $N$ agents can be explicitly determined by the $N-1$ edges of a tree spanning the network. All other edges are redundant for controlling the structure of the network. The argument in \lem{lem1} is energetic: \eqref{equ11}, which holds for the tree network $\mathcal{G}(0)$, serves to bound $\dot{V}_{1}$ in~\eqref{equ21} by~\eqref{equ25}; then, by~\prop{prop2}, the time integration~\eqref{equ26} implies that the distances between all pairs of slaves $(i,j)\in\mathcal{E}(0)$ can be kept strictly smaller than the communication radius $r$. Hence, the distributed control $\fbm{u}_{i}$ can regulate the energy $V_{p}$ stored in the network using the local $\bm{\theta}_{i}$. 

Compared to autonomous and leader-follower MRS-s, connectivity preservation for teleoperated swarms faces the novel challenge of the unpredictable user command $\fbm{f}$. Unlike external disturbances, say wind forces, the user input $\fbm{f}$ guides the motion of the slave swarm by commanding the informed slave. Hence, the user command cannot be rejected in swarm teleoperation. Nevertheless, $\fbm{f}$ may endanger the connectivity of the swarm, for example by moving the informed slave such that not all slaves can keep up with it. This section proves that the distributed control~\eqref{equ13} can be designed to eliminate the threat posed by the operator command to the connectivity of the swarm. The uniqueness of the proposed control is apparent in the condition~\eqref{equ24} on $K_{i}(t)$, especially in its P+d form in~\rem{rem6}. Namely, the control modulates dynamically and simultaneously the interconnections and local damping injection of the slaves according to their inter-distances. To the authors' best knowledge, the control~\eqref{equ13} is the first strategy to tackle connectivity-preserving swarm teleoperation with a state-dependent updating of the coupling and damping gains.

%%%%%%%%%%%%%%%%%%%%%%%%%%%%%%%%%%%%%
\subsection{Time-Delay Slave Swarm}\label{sec: time delay}
%%%%%%%%%%%%%%%%%%%%%%%%%%%%%%%%%%%%%

In a time-delay swarm, the slaves move during the time it takes their information to reach their neighbours. Thus, the distance dependency of inter-slave communications needs to include the mismatched inter-slave distances. By~\ass{ass3}, at time $t$, slave $i$ receives the information sent by slave $j$ at time $t-T_{ji}(t)$ and slave $j$ receives the information sent by slave $i$ at time $t-T_{ij}(t)$. Therefore, the bidirectional communication link $(i,j)$ between slaves $i$ and $j$ exists at time $t$, i.e., $(i,j)\in\mathcal{E}(t)$, iff $\|\fbm{x}_{i}(t)-\fbm{x}_{j}(t-T_{ji}(t))\|< r$ and $\|\fbm{x}_{j}(t)-\fbm{x}_{i}(t-T_{ij}(t))\|< r$. Letting:
\begin{equation}\label{equ27}
|\fbm{x}_{ij}(\tau)|_{T}=\sup\limits_{-\overline{T}_{ji}\leq\delta\leq 0}\|\fbm{x}_{i}(\tau)-\fbm{x}_{j}(\tau+\delta)\|\textrm{,}
\end{equation}
every interaction link $(i,j)\in\mathcal{E}(0)$ in the tree $\mathcal{G}(0)$ can be maintained by making $|\fbm{x}_{ij}(t)|_{T}<r$ and $|\fbm{x}_{ji}(t)|_{T}<r\ \forall t\geq 0$.  

As in~\sect{sec: delay free}, the control design and the proof of connectivity maintenance use set invariance analysis. Specifically: 
\begin{align*}
|\fbm{x}_{ij}(t)|_{T}=&\sup\limits_{-\overline{T}_{ji}\leq\delta\leq 0}\|\fbm{x}_{ij}(t)+\fbm{x}_{j}(t)-\fbm{x}_{j}(t+\delta)\|\\
\leq &\|\fbm{x}_{ij}(t)\|+\sup\limits_{-\overline{T}_{ji}\leq\delta\leq 0}\|\fbm{x}_{j}(t)-\fbm{x}_{j}(t+\delta)\|\textrm{,}
\end{align*}
by the triangle inequality with $\fbm{x}_{ij}(t)=\fbm{x}_{i}(t)-\fbm{x}_{j}(t)$, and thus $|\fbm{x}_{ij}(t)|_{T}<r$ is ensured by making:
\begin{align*}
\|\fbm{x}_{ij}(t)\|<\overline{r}\quad \text{and}\quad \sup\limits_{-\overline{T}_{ji}\leq\delta\leq 0}\|\fbm{x}_{j}(t)-\fbm{x}_{j}(t+\delta)\|\leq \kappa\epsilon\textrm{,}
\end{align*}
for $\overline{r}=r-\kappa\epsilon$ and $0<\kappa<1$. In other words, all interaction links in the time-delay tree network $\mathcal{G}(0)$ can be maintained by rendering invariant the following sets:
\begin{equation}\label{equ28}
\begin{aligned}
\mathcal{S}_{i}(t)=&\Bigg\{j\in\mathcal{N}_{i}(0)\ \Big|\ \|\fbm{x}_{ij}(t)\|<\overline{r}\quad \text{and}\\
&\quad \sup\limits_{-\overline{T}_{ji}\leq\delta\leq 0}\|\fbm{x}_{j}(t)-\fbm{x}_{j}(t+\delta)\|\leq \kappa\epsilon\Bigg\}
\end{aligned}
\end{equation} 
for all slaves $i=1,\cdots,N$ and for all time $t\geq 0$. Assumptions~\ref{ass1} and~\ref{ass4} imply that $\mathcal{S}_{i}(0)=\mathcal{N}_{i}(0)$. Letting:
\begin{equation}\label{equ29}
\mathcal{S}_{i}(\tau)=\mathcal{N}_{i}(0)\ \forall\tau\in[0,t)\ \textrm{and }\forall i=1,\cdots,N
\end{equation} 
 it suffices to prove that $\mathcal{S}_{i}(t)=\mathcal{N}_{i}(0)$ by induction on time. 

By the assumption~\eqref{equ29}, the control of slave $i$ at time $t$ can use the delayed information of all slaves $j\in\mathcal{N}_{i}(0)$ as follows:
\begin{equation}\label{equ30}
\fbm{u}_{i}=-\Big(\sigma K_{i}(t)+B\Big)\hfbm{\theta}_{i}-\Big(K_{i}(t)+D\Big)\dbm{x}_{i}\textrm{,}
\end{equation}
with $K_{i}(t)$, $B$, $D$, $\sigma$ positive scalars to be determined, and:
\begin{align*}
\hfbm{\theta}_{i}=\sum_{j\in\mathcal{N}_{i}(0)}\nabla_{i}\psi(\|\fbm{x}^{d}_{ij}\|)\textrm{,}
\end{align*}
where $\fbm{x}^{d}_{ij}=\fbm{x}_{i}-\fbm{x}_{jd}$, $\fbm{x}_{jd}=\fbm{x}_{j}(t-T_{ji}(t))$ and:
\begin{align*}
\nabla_{i}\psi(\|\fbm{x}^{d}_{ij}\|)=\frac{2P\cdot\left(r^{2}+Q\right)}{\left(r^{2}-\|\fbm{x}^{d}_{ij}\|^{2}+Q\right)^{2}}\left(\fbm{x}_{i}-\fbm{x}_{jd}\right)\textrm{.}
\end{align*}

The Lyapunov-Krasovskii functional candidate for the investigation of the connectivity of the time-delay swarm is:
\begin{equation}\label{equ31}
V_{2}=V_{1}+\sum^{N}_{i=1}\frac{\Omega V_{ci}+\Upsilon V_{si}}{B+\sigma D}\textrm{,}
\end{equation}
where $V_{1}$ is defined in~\eqref{equ14}, and $V_{ci}$ and $V_{si}$ measure the impact of the delays on swarm connectivity and stability by:
\begin{align*}
V_{ci}=&\sum_{j\in\mathcal{N}_{i}(0)}\overline{T}_{ji}\int^{t}_{t-\overline{T}_{ji}}e^{-\rho(t-\tau)}\|\dbm{x}_{j}(\tau)\|^{2}d\tau\textrm{,}\\
V_{si}=&\sum_{j\in\mathcal{N}_{i}(0)}\overline{T}_{ji}\int^{0}_{-\overline{T}_{ji}}\int^{t}_{t+\delta}e^{-\rho(t-\tau)}\|\dbm{x}_{j}(\tau)\|^{2}d\tau d\delta
\end{align*}
with $\rho$, $\Omega$ and $\Upsilon$ positive constants to be determined. In particular, $V_{ci}$ serves to upper-bound the mismatched distance $\|\fbm{x}_{j}(t)-\fbm{x}_{j}(t-T_{ji}(t))\|$ caused by the delay $T_{ji}(t)$, and $V_{si}$ helps determine the damping to inject to stabilize the swarm.

As in~\eqrefs{equ14}{equ22}, the derivative of $V_{1}$ along the trajectories of~\eqref{equ9} in closed-loop with~\eqref{equ30} can be upper-bounded by:
\begin{equation}\label{equ32}
\begin{aligned}
\dot{V}_{1}\leq &-\sum^{N}_{i=1}\frac{\widehat{K}_{i}(t)\tbm{s}_{i}\fbm{s}_{i}+\widehat{D}_{i}\dtbm{x}_{i}\dbm{x}_{i}+\sigma B\tfbm{\theta}_{i}\bm{\theta}_{i}}{B+\sigma D}\\
&+\sum^{N}_{i=1}\frac{\sigma K_{i}(t)+B}{B+\sigma D}\tbm{s}_{i}\tilfbm{\theta}_{i}+\frac{\|\fbm{f}\|^{2}}{B+\sigma D}\textrm{,}
\end{aligned}
\end{equation}
where $\tilfbm{\theta}_{i}=\bm{\theta}_{i}-\hfbm{\theta}_{i}$ with $\bm{\theta}_{i}$ given in~\eqref{equ7}, and $\widehat{K}_{i}(t)$, $\widehat{D}_{i}$ are defined in~\eqref{equ22}. As $\overline{T}_{ji}$ are constants, the derivative of $V_{ci}$ is:
\begin{equation}\label{equ33}
\begin{aligned}
\dot{V}_{ci}=&\sum_{j\in\mathcal{N}_{i}(0)}\overline{T}_{ji}\left(\|\dbm{x}_{j}(t)\|^{2}-e^{-\rho\overline{T}_{ji}}\left\|\dbm{x}_{j}\left(t-\overline{T}_{ji}\right)\right\|^{2}\right)\\
&-\rho V_{ci}\leq \sum_{j\in\mathcal{N}_{i}(0)}\overline{T}_{ji}\dtbm{x}_{j}\dbm{x}_{j}-\rho V_{ci}\textrm{,}
\end{aligned}
\end{equation}
and the derivative of $V_{si}$ can be bounded by:
\begin{equation}\label{equ34}
\begin{aligned}
\dot{V}_{si}=&\sum_{j\in\mathcal{N}_{i}(0)}\overline{T}_{ji} \int^{0}_{-\overline{T}_{ji}}\|\dbm{x}_{j}(t)\|^{2}-e^{\rho\delta}\|\dbm{x}_{j}(t+\delta)\|^{2} d\delta-\rho V_{si}\\
=&\sum_{j\in\mathcal{N}_{i}(0)}\overline{T}^{2}_{ji}\dtbm{x}_{j}\dbm{x}_{j}-\overline{T}_{ji}\int^{0}_{-\overline{T}_{ji}}e^{\rho\delta}\|\dbm{x}_{j}(t+\delta)\|^{2}d\delta-\rho V_{si}\\
\leq &\sum_{j\in\mathcal{N}_{i}(0)}\overline{T}^{2}_{ji}\dtbm{x}_{j}\dbm{x}_{j}-\frac{\overline{T}_{ji}}{e^{\rho\overline{T}_{ji}}}\int^{0}_{-\overline{T}_{ji}}\|\dbm{x}_{j}(t+\delta)\|^{2}d\delta-\rho V_{si}\\
\leq &\sum_{j\in\mathcal{N}_{i}(0)}\overline{T}^{2}_{ji}\dtbm{x}_{j}\dbm{x}_{j}-\frac{\overline{T}_{ji}}{e^{\rho\overline{T}_{ji}}}\int^{t}_{t-T_{ji}(t)}\|\dbm{x}_{j}(\tau)\|^{2}d\tau-\rho V_{si}\textrm{.}
\end{aligned}
\end{equation}

In~\eqref{equ32}, the term containing $\tbm{s}_{i}\tilfbm{\theta}_{i}$, due to the delay-induced distortions $\tilfbm{\theta}_{i}$, endangers the connectivity and stability of the swarm. The first step to overcome its threat is to upper-bound the impact of the delay-induced distortions $\tilfbm{\theta}_{i}$ by the local state variables:
\begin{align*}
&\tbm{s}_{i}\nabla_{i}\psi(\|\fbm{x}_{ij}\|)-\tbm{s}_{i}\nabla_{i}\psi(\|\fbm{x}^{d}_{ij}\|)\\
=&\frac{2P\cdot\left(r^{2}+Q\right)\tbm{s}_{i}\fbm{x}_{ij}}{\left(r^{2}-\|\fbm{x}_{ij}\|^{2}+Q\right)^{2}}-\frac{2P\cdot(r^{2}+Q)\tbm{s}_{i}\fbm{x}_{ij}}{\left(r^{2}-\|\fbm{x}^{d}_{ij}\|^{2}+Q\right)^{2}}\\
&-\frac{2P\cdot(r^{2}+Q)\tbm{s}_{i}(\fbm{x}_{j}-\fbm{x}_{jd})}{\left(r^{2}-\|\fbm{x}^{d}_{ij}\|^{2}+Q\right)^{2}}\\
=&\frac{2P\cdot(r^{2}+Q)\tbm{s}_{i}\fbm{x}_{ij}}{\left(r^{2}-\|\fbm{x}_{ij}\|^{2}+Q\right)^{2}\left(r^{2}-\|\fbm{x}^{d}_{ij}\|^{2}+Q\right)^{2}}\\
&\cdot\left[\left(r^{2}-\|\fbm{x}^{d}_{ij}\|^{2}+Q\right)^{2}-\left(r^{2}-\|\fbm{x}_{ij}\|^{2}+Q\right)^{2}\right]\\
&-\frac{2P\cdot(r^{2}+Q)\tbm{s}_{i}(\fbm{x}_{j}-\fbm{x}_{jd})}{\left(r^{2}-\|\fbm{x}^{d}_{ij}\|^{2}+Q\right)^{2}}\\
\leq &\frac{2P\cdot(r^{2}+Q)\left|\tbm{s}_{i}\fbm{x}_{ij}\right|\left|\|\fbm{x}_{ij}\|^{2}-\|\fbm{x}^{d}_{ij}\|^{2}\right|}{\left(r^{2}-\|\fbm{x}_{ij}\|^{2}+Q\right)\left(r^{2}-\|\fbm{x}^{d}_{ij}\|^{2}+Q\right)}\\
&\cdot\left(\frac{1}{r^{2}-\|\fbm{x}_{ij}\|^{2}+Q}+\frac{1}{r^{2}-\|\fbm{x}^{d}_{ij}\|^{2}+Q}\right)\\
&+\frac{2P\cdot(r^{2}+Q)\left|\tbm{s}_{i}(\fbm{x}_{j}-\fbm{x}_{jd})\right|}{\left(r^{2}-\|\fbm{x}^{d}_{ij}\|^{2}+Q\right)^{2}}
\end{align*}
\begin{equation}\label{equ35}
\begin{aligned}
\leq&\frac{2P\cdot(r^{2}+Q)\left|\tbm{s}_{i}\fbm{x}_{ij}\right|\left(\|\fbm{x}_{ij}\|+\|\fbm{x}^{d}_{ij}\|\right)\|\fbm{x}_{j}-\fbm{x}_{jd}\|}{\left(r^{2}-\|\fbm{x}_{ij}\|^{2}+Q\right)\left(r^{2}-\|\fbm{x}^{d}_{ij}\|^{2}+Q\right)}\\
&\cdot\left(\frac{1}{r^{2}-\|\fbm{x}_{ij}\|^{2}+Q}+\frac{1}{r^{2}-\|\fbm{x}^{d}_{ij}\|^{2}+Q}\right)\\
&+\frac{2P\cdot(r^{2}+Q)\|\fbm{s}_{i}\|\|\fbm{x}_{j}-\fbm{x}_{jd}\|}{\left(r^{2}-\|\fbm{x}^{d}_{ij}\|^{2}+Q\right)^{2}}\\
=&\Lambda_{ij2}(t)\|\fbm{s}_{i}\|\|\fbm{x}_{j}-\fbm{x}_{jd}\|\textrm{,}
\end{aligned}
\end{equation}
where the assumption~\eqref{equ29} has permitted to use $\|\fbm{x}_{ij}(t)\|<r$ and $\|\fbm{x}^{d}_{ij}\|<r$, and:
\begin{align*}
\Lambda_{ij2}(t)=&\frac{2P\cdot(r^{2}+Q)\left(\|\fbm{x}_{ij}\|+\|\fbm{x}^{d}_{ij}\|\right)\|\fbm{x}_{ij}\|}{\left(r^{2}-\|\fbm{x}_{ij}\|^{2}+Q\right)\left(r^{2}-\|\fbm{x}^{d}_{ij}\|^{2}+Q\right)}\\
&\cdot\left(\frac{1}{r^{2}-\|\fbm{x}_{ij}\|^{2}+Q}+\frac{1}{r^{2}-\|\fbm{x}^{d}_{ij}\|^{2}+Q}\right)\\
&+\frac{2P\cdot(r^{2}+Q)}{\left(r^{2}-\|\fbm{x}^{d}_{ij}\|^{2}+Q\right)^{2}}\textrm{.}
\end{align*}

Employing~\eqref{equ35} and the Cauchy-Schwarz inequality, the stability threat due to the time delays can be limited by:
\begin{equation}\label{equ36}
\begin{aligned}
&\Big(\sigma K_{i}(t)+B\Big)\tbm{s}_{i}\tilfbm{\theta}_{i}-\sum_{j\in\mathcal{N}_{i}(0)}\frac{\Upsilon\overline{T}_{ji}}{e^{\rho\overline{T}_{ji}}}\int^{t}_{t-T_{ji}(t)}\|\dbm{x}_{j}(\tau)\|^{2}d\tau\\
\leq &\Big(\sigma K_{i}(t)+B\Big)\sum_{j\in\mathcal{N}_{i}(0)}\Lambda_{ij2}(t)\|\fbm{s}_{i}\|\|\fbm{x}_{j}-\fbm{x}_{jd}\|\\
&-\Upsilon\sum_{j\in\mathcal{N}_{i}(0)}e^{-\rho\overline{T}_{ji}}\overline{T}_{ji}\sum^{n}_{k=1}\int^{t}_{t-T_{ji}(t)}|\dot{x}^{k}_{j}(\tau)|^{2}d\tau\\
\leq &\Big(\sigma K_{i}(t)+B\Big)\sum_{j\in\mathcal{N}_{i}(0)}\Lambda_{ij2}(t)\|\fbm{s}_{i}\|\|\fbm{x}_{j}-\fbm{x}_{jd}\|\\
&-\Upsilon\sum_{j\in\mathcal{N}_{i}(0)}e^{-\rho\overline{T}_{ji}}\sum^{n}_{k=1}\left|\int^{t}_{t-T_{ji}(t)}|\dot{x}^{k}_{j}(\tau)|d\tau\right|^{2}\\
\leq &\Big(\sigma K_{i}(t)+B\Big)\sum_{j\in\mathcal{N}_{i}(0)}\Lambda_{ij2}(t)\|\fbm{s}_{i}\|\|\fbm{x}_{j}-\fbm{x}_{jd}\|\\
&-\Upsilon\sum_{j\in\mathcal{N}_{i}(0)}e^{-\rho\overline{T}_{ji}}\sum^{n}_{k=1}\left|\int^{t}_{t-T_{ji}(t)}\dot{x}^{k}_{j}(\tau)d\tau\right|^{2}\\
%=&H_{i}(t)\sum_{j\in\mathcal{N}_{i}(0)}\Lambda_{ij2}(t)\|\fbm{s}_{i}\|\|\fbm{x}_{j}-\fbm{x}_{jd}\|\\
%&-\Upsilon\sum_{j\in\mathcal{N}_{i}(0)}e^{-\rho\overline{T}_{ji}}\sum^{n}_{k=1}\left|x^{k}_{j}(t)-x^{k}_{j}(t-T_{ji}(t))\right|^{2}\\
=&\Big(\sigma K_{i}(t)+B\Big)\sum_{j\in\mathcal{N}_{i}(0)}\Lambda_{ij2}(t)\|\fbm{s}_{i}\|\|\fbm{x}_{j}-\fbm{x}_{jd}\|\\
&-\Upsilon\sum_{j\in\mathcal{N}_{i}(0)}e^{-\rho\overline{T}_{ji}}\|\fbm{x}_{j}-\fbm{x}_{jd}\|^{2}\\
\leq &\frac{1}{4\Upsilon}\left[\sigma K_{i}(t)+B\right]^{2}\sum_{j\in\mathcal{N}_{i}(0)}\Lambda^{2}_{ij2}(t) e^{\rho\overline{T}_{ji}}\|\fbm{s}_{i}\|^{2}\textrm{,}
\end{aligned}
\end{equation}
where $\dot{x}^{k}_{j}(\tau)$ is the $k$-th element of $\dbm{x}_{j}(\tau)$, $k=1,\cdots,n$. Substitution from~\eqref{equ36} in the combined~\eqrefs{equ32}{equ34} bounds the derivative of $V_{2}$ by:
\begin{equation}\label{equ37}
\begin{aligned}
\dot{V}_{2}\leq &-\sum^{N}_{i=1}\frac{\widetilde{K}_{i}(t)\tbm{s}_{i}\fbm{s}_{i}+\widetilde{D}_{i}\dtbm{x}_{i}\dbm{x}_{i}+\rho\cdot(\Omega V_{ci}+\Upsilon V_{si})}{B+\sigma D}\\
&-\frac{\sigma B}{B+\sigma D} \sum^{N}_{i=1}\tfbm{\theta}_{i}\bm{\theta}_{i}+\frac{\|\fbm{f}\|^{2}}{B+\sigma D}\textrm{,}
\end{aligned}
\end{equation}
where:
\begin{align*}
\widetilde{K}_{i}(t)=&\widehat{K}_{i}(t)-\frac{1}{4\Upsilon}\left[\sigma K_{i}(t)+B\right]^{2}\sum_{j\in\mathcal{N}_{i}(0)}\Lambda^{2}_{ij2}(t) e^{\rho\overline{T}_{ji}}\textrm{,}\\
\widetilde{D}_{i}=&\widehat{D}_{i}-\sum_{j\in\mathcal{N}_{i}(0)}\left(\Upsilon\overline{T}^{2}_{ij}+\Omega\overline{T}_{ij}\right)\textrm{.}
\end{align*}

%%%%%%%%%
\begin{theorem}\label{theorem2}
The control~\eqref{equ30} renders the sets $\mathcal{S}_{i}(t),\ i=1,\cdots,N$, invariant and thus maintains the connectivity of the teleoperated time-delay swarm~\eqref{equ1} with \asses{ass1}{ass4} if it satisfies~\eqref{equ4} and:
\begin{align}\label{equ38}
\widetilde{K}_{i}(t)\geq \frac{1}{2}\rho\lambda_{i2}\quad \text{and}\quad \widetilde{D}_{i}\geq 0\textrm{,}
\end{align}
for all $i=1,\cdots,N$, where $\Delta$ is given in~\theo{theorem1} and:
\begin{equation}\label{equ39}
\begin{aligned}
\rho=\frac{4\sigma\lambda_{L}PB}{(r^{2}+Q)(B+\sigma D)}\text{ and }\Omega=\frac{P\overline{r}^{2} e^{\rho\overline{T}}(B+\sigma D)}{\kappa^{2}\epsilon^{2}\cdot\left(r^{2}-\overline{r}^{2}+Q\right)}
\end{aligned}
\end{equation}
with $\overline{T}=\max\limits_{(i,j)}\left(\overline{T}_{ij}\right)$ the maximum delay in the network.
\end{theorem}
\begin{proof}
As in the proof of~\theo{theorem1}, \eq{equ11} and condition~\eqref{equ38} yield:
\begin{align*}
V_{2}(t)\leq e^{-\rho t} V_{2}(0)+\frac{\overline{f}^{2}}{\rho\cdot\left(B+\sigma D\right)}\textrm{.}
\end{align*}
Assumptions~\ref{ass3} and~\ref{ass4} indicate that $V_{ci}(0)=V_{si}(0)=0$. Because $V_{ci}(t)$ and $V_{si}(t)$ are nonnegative, it follows that:
\begin{align*}
V_{p}(t)\leq V_{2}(t)\leq V_{2}(0)+\frac{\overline{f}^{2}}{\rho\cdot\left(B+\sigma D\right)}=V_{p}(0)+\Delta\textrm{,}
\end{align*}
which, together with~\prop{prop2}, implies that $\|\fbm{x}_{ij}(t)\|<\overline{r}$ for all $(i,j)\in\mathcal{E}(0)$.

The Cauchy-Schwarz inequality leads to:
\begin{align*}
&\overline{T}_{ji}\int^{t}_{t-\overline{T}_{ji}}\|\dbm{x}_{j}(\tau)\|^{2}d\tau\\ 
=&\overline{T}_{ji}\sum^{n}_{k=1}\int^{t}_{t-\overline{T}_{ji}}\left|\dot{x}^{k}_{j}(\tau)\right|^{2}d\tau \geq \sum^{n}_{k=1}\left|\int^{t}_{t-\overline{T}_{ji}}\left|\dot{x}^{k}_{j}(\tau)\right|d\tau\right|^{2}\\
\geq &\sum^{n}_{k=1}\left|\sup\limits_{-\overline{T}_{ji}\leq\delta\leq 0}x^{k}_{j}(t+\delta)-\inf\limits_{-\overline{T}_{ji}\leq\delta\leq 0}x^{k}_{j}(t+\delta)\right|^{2}\\
\geq &\sup\limits_{-\overline{T}_{ji}\leq\delta_{1},\delta_{2}\leq 0}\|\fbm{x}_{j}(t+\delta_{1})-\fbm{x}_{j}(t+\delta_{2})\|^{2}\\
\geq &\sup\limits_{-\overline{T}_{ji}\leq\delta\leq 0}\|\fbm{x}_{j}(t)-\fbm{x}_{j}(t+\delta)\|^{2}\textrm{,}
\end{align*}
and to the lower-bound on $V_{2}(t)$:
\begin{align*}
V_{2}(t)\geq &\frac{\Omega V_{ci}}{B+\sigma D}=\frac{\Omega\overline{T}_{ji}}{B+\sigma D}\int^{t}_{t-\overline{T}_{ji}}e^{-\rho(t-\tau)}\|\dbm{x}_{j}(\tau)\|^{2}d\tau\\
\geq &\frac{\Omega\overline{T}_{ji}e^{-\rho\overline{T}_{ji}}}{B+\sigma D}\int^{t}_{t-\overline{T}_{ji}}\|\dbm{x}_{j}(\tau)\|^{2}d\tau\\ 
\geq &\frac{\Omega e^{-\rho\overline{T}_{ji}}}{B+\sigma D}\sup\limits_{-\overline{T}_{ji}\leq\delta\leq 0}\|\fbm{x}_{j}(t)-\fbm{x}_{j}(t+\delta)\|^{2}\textrm{.}
\end{align*}
Since $V_{2}(t)$ satisfies~\eqref{equ4}, $V_{2}(t)\leq V_{p}(0)+\Delta<\psi_{\max}$, condition~\eqref{equ39} further leads to:
\begin{align*}
&\sup\limits_{-\overline{T}_{ji}\leq\delta\leq 0}\|\fbm{x}_{j}(t)-\fbm{x}_{j}(t+\delta)\|\\
\leq &\sqrt{\frac{e^{\rho\overline{T}_{ji}}}{\Omega}\left(B+\sigma D\right)V_{2}(t)}
\leq \sqrt{\frac{P\overline{r}^{2}e^{\rho\overline{T}_{ji}}(B+\sigma D)}{\Omega\left(r^{2}-\overline{r}^{2}+Q\right)}}\leq \kappa\epsilon\textrm{.}
\end{align*}
Thus, $\mathcal{S}_{i}(t)=\mathcal{N}_{i}(0)$ are positively invariant for all $i=1,\cdots,N$, i.e., all interaction links and, with them, the connectivity of the time-delay EL tree network are preserved.
\end{proof}

The following parameter selection procedure proves the feasibility of \theo{theorem2}. First, set an upper bound $\overline{\sigma}>0$ to select $\sigma$. Then, conditions~\eqref{equ22} and~\eqref{equ39} lead to:
\begin{equation}\label{equ40}
\begin{aligned}
\widetilde{D}_{i}=&\left(1-\frac{P\overline{r}^{2}e^{\rho\overline{T}}\sigma}{\kappa^{2}\epsilon^{2}\cdot(r^{2}-\overline{r}^{2}+Q)}\sum_{j\in\mathcal{N}_{i}(0)}\overline{T}_{ij}\right)D\\
&-\sum_{j\in\mathcal{N}_{i}(0)}\left(2\sigma\cdot(\eta_{i}+\gamma_{i}+\zeta_{i}+\eta_{j}+\gamma_{j})+\Upsilon\overline{T}^{2}_{ij}\right)\\
&-\frac{PB\overline{r}^{2}e^{\rho\overline{T}}}{\kappa^{2}\epsilon^{2}\cdot(r^{2}-\overline{r}^{2}+Q)}\sum_{j\in\mathcal{N}_{i}(0)}\overline{T}_{ij}\textrm{,}
\end{aligned}
\end{equation}
and condition $\widetilde{D}_{i}\geq 0$ in~\eqref{equ38} requires that:
\begin{align*}
D>\frac{PB\overline{r}^{2}}{\kappa^{2}\epsilon^{2}\cdot(r^{2}-\overline{r}^{2}+Q)}\sum_{j\in\mathcal{N}_{i}(0)}\overline{T}_{ij}\textrm{.}
\end{align*}
Now, $\rho$ in~\eqref{equ39} can be upper-bounded by:
\begin{align*}
\rho\leq\frac{4\lambda_{L}PB}{(r^{2}+Q)D}\leq\frac{4\lambda_{L}\kappa^{2}\epsilon^{2}\cdot(r^{2}-\overline{r}^{2}+Q)}{\overline{r}^{2}(r^{2}+Q)\sum_{j\in\mathcal{N}_{i}(0)}\overline{T}_{ij}}=\overline{\rho}_{i}\textrm{.}
\end{align*}
After choosing $B>0$ heuristically and $0<\kappa<1$ and $Q>0$ small enough to guarantee $\omega_{1}>0$ and $\omega_{2}>0$ in~\eqref{equ4}, condition~\eqref{equ39} upper-bounds $\Delta$ by:
\begin{align*}
\Delta\leq\frac{\sigma^{2}}{B}\sum^{N}_{i=1}\lambda_{i2}\|\bm{\theta}_{i}(0)\|^{2}+\frac{(r^{2}+Q)\overline{f}^{2}}{4\sigma\lambda_{L}PB}=\overline{\Delta}
\textrm{.}
\end{align*}
Selecting a sufficiently large $P>0$ and small $0<\sigma<\overline{\sigma}$ by:
\begin{equation}\label{equ41}
\begin{aligned}
\sigma P\overline{r}^{2}e^{\overline{\rho}_{i}\overline{T}}\sum_{j\in\mathcal{N}_{i}(0)}\overline{T}_{ij}<\kappa^{2}\epsilon^{2}\cdot(r^{2}-\overline{r}^{2}+Q)\textrm{,}\\
P\omega_{2}>(r^{2}-\overline{r}^{2}+Q)(r^{2}-\hat{r}^{2}+Q)\overline{\Delta}
\end{aligned}
\end{equation}
then permits to select $D$ and makes $\omega_{3}>0$ in~\eqref{equ4}.

Because $\Lambda_{ij1}(t)$ and $\Lambda_{ij2}(t)$ depend on $\|\fbm{x}_{ij}(t)\|$, the updating of $K_{i}(t)$ by~\eqref{equ38} requires the neighbours' non-delayed information. Nevertheless, the triangle inequality:
\begin{align*}
\|\fbm{x}_{ij}(t)\|\leq \|\fbm{x}^{d}_{ij}(t)\|+\|\fbm{x}_{j}(t)-\fbm{x}_{j}(t-T_{ji}(t))\|
\end{align*}
permits to replace $\|\fbm{x}_{ij}\|$ with:
\begin{align}\label{equ42}
\chi_{ij}(t)=\min\left(\overline{r},\|\fbm{x}^{d}_{ij}\|+\kappa\epsilon\right)
\end{align}
when updating $K_{i}(t)$. More specifically, the invariance of $\mathcal{S}_{i}(\tau)$ on $\tau\in[0,t)$, presumed in~\eqref{equ29}, indicates that $\|\fbm{x}_{ij}(t)\|\leq \chi_{ij}(t)$. After choosing $\eta_{i}$, $\gamma_{i}$ and $\zeta_{i}$ heuristically, an upper-bound on $\Lambda_{ij1}(t)$ in~\eqref{equ19} is given by:
\begin{equation}\label{equ43}
\begin{aligned}
\overline{\Lambda}_{ij1}(t)=&\frac{P^{2}(r^{2}+Q)}{(r^{2}-\chi^{2}_{ij}(t)+Q)^{4}}\left(\frac{\lambda^{2}_{i2}}{\gamma_{i}}+\frac{c^{2}_{i}\chi^{2}_{ij}(t)}{2\zeta_{i}}\right)\\
&+\frac{16\lambda^{2}_{i2}P^{2}\chi^{2}_{ij}(t)(r^{2}+Q)^{2}}{\eta_{i}\cdot(r^{2}-\chi^{2}_{ij}(t)+Q)^{6}}\textrm{.}
\end{aligned}
\end{equation}
Because $\widetilde{K}_{i}(t)$ is quadratic in $K_{i}(t)$, the key to making $\widetilde{K}_{i}(t)\geq\rho\lambda_{i2}/2$ is to reduce the order of $\widetilde{K}_{i}(t)$ in $K_{i}(t)$. To this end, first note that the assumption~\eqref{equ29} leads to $\overline{\Lambda}_{ij1}(t)\leq\overline{\Lambda}_{ij1}$ and to $\Lambda_{ij2}(t)\leq\overline{\Lambda}_{ij2}$ with:
\begin{equation}\label{equ44}
\begin{aligned}
\overline{\Lambda}_{ij1}=&\frac{\lambda^{2}_{i2}P^{2}\cdot(r^{2}+Q)^{2}}{\gamma_{i}\cdot(r^{2}-\overline{r}^{2}+Q)^{4}}+\frac{c^{2}_{i}P^{2}\overline{r}^{2}(r^{2}+Q)}{2\zeta_{i}\cdot(r^{2}-\overline{r}^{2}+Q)^{4}}\\
&+\frac{16\lambda^{2}_{i2}P^{2}\overline{r}^{2}(r^{2}+Q)^{2}}{\eta_{i}\cdot(r^{2}-\overline{r}^{2}+Q)^{6}}\textrm{,}\\
\overline{\Lambda}_{ij2}=&\frac{2P\overline{r}(r^{2}+Q)(\overline{r}+r)}{(r^{2}-\overline{r}^{2}+Q)Q}\left(\frac{1}{r^{2}-\overline{r}^{2}+Q}+\frac{1}{Q}\right)\\
&+\frac{2P\cdot(r^{2}+Q)}{Q^{2}}\textrm{.}
\end{aligned}
\end{equation}
With the $0<\sigma<\overline{\sigma}$, $P$, $Q$ and $B$ chosen above, set the following upper bound on $K_{i}(t)$ for every slave $i$:
\begin{equation}\label{equ45}
\begin{aligned}
\overline{K}_{i}=\frac{\overline{\sigma}}{\overline{\sigma}-\sigma}\left(\sum_{j\in\mathcal{N}_{i}(0)}\sigma\overline{\Lambda}_{ij1}+\frac{\varrho_{i}}{4}+\frac{\overline{\rho}_{i}\lambda_{i2}}{2}+\frac{B}{\overline{\sigma}}\right)\textrm{,}
\end{aligned}
\end{equation}
and select $\Upsilon$ by:
\begin{equation}\label{equ46}
\Upsilon=\frac{\overline{\sigma}}{4}\max\limits_{i}\left(\sum_{j\in\mathcal{N}_{i}(0)}\overline{\Lambda}^{2}_{ij2}e^{\overline{\rho}_{i}\overline{T}_{ji}}\left(\sigma\overline{K}_{i}+B\right)\right)\textrm{.}
\end{equation}
Then, the dynamic updating of $K_{i}(t)$ by:
\begin{align}\label{equ47}
K_{i}(t)=\frac{\overline{\sigma}}{\overline{\sigma}-\sigma}\left(\sum_{j\in\mathcal{N}_{i}(0)}\sigma\overline{\Lambda}_{ij1}(t)+\frac{\varrho_{i}}{4}+\frac{\overline{\rho}_{i}\lambda_{i2}}{2}+\frac{B}{\overline{\sigma}}\right)
\end{align}
guarantees that $K_{i}(t)\leq\overline{K}_{i}$ and further that:
\begin{align*}
\frac{1}{4\Upsilon}\Big(\sigma K_{i}(t)+B\Big)^{2}\sum_{j\in\mathcal{N}_{i}(0)}\Lambda^{2}_{ij2}(t)e^{\rho\overline{T}_{ji}}\leq \frac{1}{\overline{\sigma}}\Big(\sigma K_{i}(t)+B\Big)\textrm{.}
\end{align*}
Therefore, $\widetilde{K}_{i}(t)$ can be lower-bounded by:
\begin{align*}
\widetilde{K}_{i}(t)\geq &\widehat{K}_{i}(t)-\frac{1}{\overline{\sigma}}\Big(\sigma K_{i}(t)+B\Big)\\
=&\frac{\overline{\sigma}-\sigma}{\overline{\sigma}}K_{i}(t)-\sum_{j\in\mathcal{N}_{i}(0)}\sigma\Lambda_{ij1}(t)-\frac{\varrho_{i}}{4}-\frac{B}{\overline{\sigma}}\geq\frac{\rho\lambda_{i2}}{2}\textrm{.}
\end{align*}

Finally, the $\Upsilon$ in~\eqref{equ46} and a sufficiently large damping gain $D$ guarantee that $\widetilde{D}_{i}\geq 0$ in~\eqref{equ40}. Algorithm~\ref{alg1} summarizes the gain selection procedure for the control~\eqref{equ30}, which maintains the tree network of the time-delay EL swarm.
\begin{algorithm}
  \caption{Gain selection for the dynamic control~\eqref{equ30}}
  \Input \\
  \Output 
  \begin{algorithmic}[1]\label{algorithm1}
	\STATE{Choose $\overline{\sigma}>0$ and $B>0$ heuristically;}
	\STATE{Select $0<\kappa<1$ and $Q>0$ such that $\omega_{1}>0$ and $\omega_{2}>0$ in condition~\eqref{equ4};}
	\STATE{Compute $\overline{\rho}_{i}$, choose $P>0$ and $0<\sigma<\overline{\sigma}$ by~\eqref{equ41} such that $\omega_{3}>0$ in condition~\eqref{equ4};}
	\STATE{Choose positive $\eta_{i}$, $\gamma_{i}$ and $\zeta_{i}$ heuristically;}
	\STATE{Set $\overline{\Lambda}_{ij1}(t)$ by~\eqref{equ43}, and update $K_{i}(t)$ dynamically by~\eqref{equ47};}
    \STATE{Compute $\overline{\Lambda}_{ij1}$, $\overline{\Lambda}_{ij2}$, $\overline{K}_{i}$ and $\Upsilon$ by~\eqref{equ44}-\eqref{equ46}, respectively;}
    \STATE{Select $D$ satisfying $\widetilde{D}_{i}\geq 0$ in~\eqref{equ40};}
    \RETURN $K_{i}(t)$, $B$ and $D$.
  \end{algorithmic}
  \label{alg1}
\end{algorithm}

%%%%%%%%%%%%%%%%%%%%%%%%%%%%%%%%%%%%%
\subsection{ISS Swarm Teleoperation}\label{sec: iss}

Connectivity-preserving swarm teleoperation control permits the human operator to command the motion of the slave swarm subject to maintaining its network connectivity~\cite{Lee2013TMECH}. Similarly to~\cite{Lee2010TRO}, this paper defines the state of the slave swarm by $\bm{\phi}=(\dtbm{x}, \ \ttilbm{x})^\mathsf{T}$, where: 
\begin{align*}
\dbm{x}=(\dtbm{x}_{1},\cdots,\dtbm{x}_{N})^\mathsf{T}
\text{ and } \tilbm{x}=\left(\tbm{D}\otimes\fbm{I}_{n}\right)\fbm{x}
\end{align*}
stack the slave velocities $\dbm{x}_{i}$ and the position errors $\fbm{x}_{ij}=\fbm{x}_{i}-\fbm{x}_{j}$ between all pairs of adjacent slaves $(i,j)\in\mathcal{E}(0)$. Because \theo{theorem2} has proven connectivity preservation under the control~\eqref{equ30}, this section guarantees stable synchronization of the time-delay slave swarm by rendering the set $\mathcal{I}=\left\{\bm{\phi}\ \big|\ \bm{\phi}\in\mathcal{L}_{\infty}\right\}$ invariant and by making the set:
\begin{align*}
\mathcal{A}=\left\{\bm{\phi}\ \Big|\ \|\bm{\phi}\|\leq\alpha\left(\sup\limits_{0\leq\tau\leq t}\|\fbm{f}(\tau)\|\right)\right\}\textrm{,}
\end{align*} 
globally attractive, where $\alpha(\cdot)$ is of class $\mathcal{K}_{\infty}$.  

By~\defi{iss}, making the teleoperated time-delay swarm ISS ensures the two sets $\mathcal{I}$ and $\mathcal{A}$ by Corollary~\ref{cor1} below.
%%%%%%%%
\begin{corollary}\label{cor1}
The time-delay slave swarm~\eqref{equ1}, with~\asses{ass1}{ass4} and in closed-loop with the control~\eqref{equ30} with parameters $\overline{\sigma}$, $\sigma$, $\eta_{i}$, $\gamma_{i}$, $\zeta_{i}$, $P$ and $Q$, and gains $K_{i}(t)$, $B$ and $D$ selected as in~\theo{theorem2}, is ISS with input $\fbm{f}$ and state $\bm{\phi}$.
\end{corollary}   
\begin{proof}
Let $\overline{\lambda}_{L}$ be the maximum eigenvalue of $\fbm{L}_{e}$. Then \eqref{equ12} implies that:
\begin{align*}
&\sum^{N}_{i=1}\tfbm{\theta}_{i}\bm{\theta}_{i}\leq \overline{\lambda}_{L}\cdot\overline{\fbm{x}}^\mathsf{T}\overline{\fbm{x}}=\frac{\overline{\lambda}_{L}}{2}\sum^{N}_{i=1}\sum_{j\in\mathcal{N}_{i}(0)}\|\nabla_{i}\psi(\|\fbm{x}_{ij}\|)\|^{2}\\
=&\sum^{N}_{i=1}\sum_{j\in\mathcal{N}_{i}(0)}\frac{2\overline{\lambda}_{L}P\cdot(r^{2}+Q)^{2}}{(r^{2}-\|\fbm{x}_{ij}\|^{2}+Q)^{3}}\psi(\|\fbm{x}_{ij}\|)\leq\frac{4\overline{\lambda}_{L}P}{r^{2}+Q}V_{p}\textrm{,}
\end{align*}
and the definition~\eqref{equ6} of $\fbm{s}_{i}$ leads to:
\begin{align*}
\sum^{N}_{i=1}\dtbm{x}_{i}\dbm{x}_{i}\leq &2\sum^{N}_{i=1}(\tbm{s}_{i}\fbm{s}_{i}+\sigma^{2}\tfbm{\theta}_{i}\bm{\theta}_{i})\leq \sum^{N}_{i=1}2\tbm{s}_{i}\fbm{s}_{i}+\frac{8\sigma^{2}\overline{\lambda}_{L}P}{r^{2}+Q}V_{p}\textrm{.}
\end{align*}
Thus, the Lyapunov-Krasovskii functional candidate $V_{2}$ can be lower-bounded by:
\begin{equation}\label{equ48}
\begin{aligned}
V_{2}\geq & V_{1}\geq \frac{\lambda_{1}}{2(B+\sigma D)}\sum^{N}_{i=1}\tbm{s}_{i}\fbm{s}_{i}+V_{p}\\
\geq &\min\Bigg(\frac{\lambda_{1}}{4(B+\sigma D)},\frac{r^{2}+Q}{8\sigma^{2}\overline{\lambda}_{L}P}\Bigg)\sum^{N}_{i=1}\dtbm{x}_{i}\dbm{x}_{i}
\end{aligned}
\end{equation}
where $\lambda_{1}=\min\limits_{i}(\lambda_{i1})$. Further, $\psi(\|\fbm{x}_{ij}\|)$ also implies that:
\begin{align}\label{equ49}
V_{2}\geq V_{p}\geq\sum^{N}_{i=1}\sum_{j\in\mathcal{N}_{i}(0)}\frac{P\|\fbm{x}_{ij}\|^{2}}{2(r^{2}+Q)}=\frac{P}{r^{2}+Q}\ttilbm{x}\tilbm{x}\textrm{.}
\end{align}
Together, \eqrefs{equ48}{equ49} lower-bound $V_{2}$ by:
\begin{equation}\label{equ50}
V_{2}\geq a_{1}\sum^{N}_{i=1}\dtbm{x}_{i}\dbm{x}_{i}+a_{1}\ttilbm{x}\tilbm{x}=a_{1}\|\bm{\phi}\|^{2}\textrm{,}
\end{equation}
where:
\begin{align*}
a_{1}=\min\Bigg(\frac{\lambda_{1}}{8(B+\sigma D)},\frac{r^{2}+Q}{16\sigma^{2}\overline{\lambda}_{L}P},\frac{P}{2(r^{2}+Q)}\Bigg)\textrm{.}
\end{align*}

From $\|\dbm{x}_{j}(\tau)\|\leq |\dbm{x}_{j}|_{\overline{T}}\ \forall \tau\in[t-\overline{T},t]$, it follows that:
\begin{align*}
V_{ci}\leq \sum_{j\in\mathcal{N}_{i}(0)}\overline{T}^{2}_{ji}|\dbm{x}_{j}|^{2}_{\overline{T}}\quad \text{and}\quad V_{si}\leq \frac{1}{2}\sum_{j\in\mathcal{N}_{i}(0)}\overline{T}^{3}_{ji}|\dbm{x}_{j}|^{2}_{\overline{T}}\textrm{.}
\end{align*}
The definition~\eqref{equ6}  of the sliding surface $\fbm{s}_{i}$ also implies that:
\begin{align*}
\sum^{N}_{i=1}\tbm{s}_{i}\fbm{s}_{i}\leq &2\sum^{N}_{i=1}(\dtbm{x}_{i}\dbm{x}_{i}+\sigma^{2}\tfbm{\theta}_{i}\bm{\theta}_{i})\leq\sum^{N}_{i=1}2\dtbm{x}_{i}\dbm{x}_{i}+\frac{8\sigma^{2}\overline{\lambda}_{L}P}{r^{2}+Q}V_{p}\textrm{.}
\end{align*}
Therefore, the Lyapunov-Krasovskii functional $V_{2}$ can be upper-bounded by:
\begin{equation}\label{equ51}
\begin{aligned}
V_{2}\leq &\sum^{N}_{i=1}\frac{\lambda_{2}\tbm{s}_{i}\fbm{s}_{i}}{2(B+\sigma D)}+\sum^{N}_{i=1}\sum_{j\in\mathcal{N}_{i}(0)}\frac{2\Omega\overline{T}^{2}_{ji}+\Upsilon\overline{T}^{3}_{ji}}{2(B+\sigma D)}|\dbm{x}_{j}|^{2}_{\overline{T}}+V_{p}\\
\leq &\sum^{N}_{i=1}\frac{\lambda_{2}\dtbm{x}_{i}\dbm{x}_{i}}{B+\sigma D}+\sum^{N}_{i=1}\sum_{j\in\mathcal{N}_{i}(0)}\frac{2\Omega\overline{T}^{2}_{ij}+\Upsilon\overline{T}^{2}_{ij}}{2(B+\sigma D)}|\dbm{x}_{i}|^{2}_{\overline{T}}\\
&+\Bigg(\frac{4\sigma^{2}\lambda_{2}\overline{\lambda}_{L}P}{(B+\sigma D)(r^{2}+Q)}+1\Bigg)\cdot\frac{P}{2Q}\sum^{N}_{i=1}\sum_{j\in\mathcal{N}_{i}(0)}\|\fbm{x}_{ij}\|^{2}\\
\leq &a_{2}\sum^{N}_{i=1}|\dbm{x}_{i}|^{2}_{\overline{T}}+\frac{a_{2}}{2}\sum^{N}_{i=1}\sum_{j\in\mathcal{N}_{i}(0)}\|\fbm{x}_{ij}\|^{2}\leq a_{2}|\bm{\phi}|^{2}_{\overline{T}}\textrm{,}
\end{aligned}
\end{equation} 
where $\lambda_{2}=\max\limits_{i}(\lambda_{i2})$ and:
\begin{align*}
a_{2}=\max\Bigg(\frac{\lambda_{2}}{B+\sigma D}+\sum_{j\in\mathcal{N}_{i}(0)}\frac{2\Omega\overline{T}^{2}_{ij}+\Upsilon\overline{T}^{2}_{ij}}{2(B+\sigma D)},\\
\frac{4\sigma^{2}\lambda_{2}\overline{\lambda}_{L}P^{2}}{(B+\sigma D)(r^{2}+Q)Q}+\frac{P}{Q}\Bigg)\textrm{.}
\end{align*}

Let $|\bm{\phi}_{t}|_{a}=\sqrt{V_{2}}$ with $\gamma_{a}=\sqrt{a_{1}}$ and $\overline{\gamma}_{a}=\sqrt{a_{2}}$. Then, $\alpha_{1}(\|\bm{\phi}\|)=a_{1}\|\bm{\phi}\|^{2}$ and $\alpha_{2}(|\bm{\phi}_{t}|_{a})=|\bm{\phi}_{t}|^{2}_{a}$ of class $\mathcal{K}_{\infty}$ trivially satisfy condition a) of Lemma~\ref{L4} in~\cite{Jiang2006SCL}.

Choosing control gains by~\theo{theorem2} simplifies $\dot{V}_{2}$ in~\eqref{equ37} to:
\begin{align*}
\dot{V}_{2}\leq &-\rho\sum^{N}_{i=1}\frac{\lambda_{i2}\tbm{s}_{i}\fbm{s}_{i}}{2(B+\sigma D)}-\rho\sum^{N}_{i=1}\frac{\Omega V_{ci}+\Upsilon V_{si}}{B+\sigma D}\\
&-\frac{(r^{2}+Q)\rho}{4\lambda_{L}P}\sum^{N}_{i=1}\tfbm{\theta}_{i}\bm{\theta}_{i}+\frac{\|\fbm{f}\|^{2}}{B+\sigma D}\leq -\rho V_{2}+\frac{\|\fbm{f}\|^{2}}{B+\sigma D}\textrm{.}
\end{align*} 
Finally, letting $\alpha_{3}(|\bm{\phi}_{t}|_{a})=\rho\cdot |\bm{\phi}_{t}|^{2}_{a}/2$ and:
\begin{align*}
\alpha_{4}(\|\fbm{f}\|)=\sqrt{\frac{2\|\fbm{f}\|^{2}}{\rho\cdot(B+\sigma D)}}
\end{align*}
of class $\mathcal{K}$ guarantees condition b) of Lemma~\ref{L4} in~\cite{Jiang2006SCL}, i.e., $\dot{V}_{2}\leq -\alpha_{3}(|\bm{\phi}_{t}|_{a})$ for every $|\bm{\phi}_{t}|_{a}\geq\alpha_{4}(\|\fbm{f}\|)$, and completes the proof that the time-delay swarm~\eqref{equ1} under the control~\eqref{equ30} is ISS with input $\fbm{f}$ and state $\bm{\phi}$. 
\end{proof}

The unpredictable user command $\fbm{f}$ makes a teleoperated swarm a perturbed nonlinear system with uncertain human operator dynamics, and makes input-to-state stability particularly useful for studying its robust synchronization. Compared to~\cite{Liu2018TIE}, which also investigates the input-to-state stability of semi-autonomous bilateral teleoperation, this paper maintains the connectivity of a teleoperated swarm with a tree network. To preserve connectivity, the controller must constrain the inter-slave distances $\|\fbm{x}_{ij}\|$ and, with them, the state $\bm{\phi}$ of the swarm below a prescribed threshold under the perturbation of user command $\fbm{f}$. Thus, a key contribution of the dynamic coupling and damping injection strategy in this paper compared to~\cite{Liu2018TIE} is to confine the impact of the user command on the connectivity of the time-delay slave swarm to a safe domain.

%%%%%%%%%%%%%%%%%%%%%%%%%%%%%%%%
\section{Experiments}\label{sec: experiments}
%%%%%%%%%%%%%%%%%%%%%%%%%%%%%%%%

This section validates the proposed dynamic strategy through experimental comparison to virtual point-based control~\cite{Lee2013TMECH}. As show in \fig{fig2}, the experimental distributed swarm teleoperation testbed includes a Novint Falcon haptic robot as the master and three Geomagic Touch haptic robots as the slave swarm. Each robot is connected to, and controlled by, a local computer running the Robot Operating System~(ROS) and MATLAB/Simulink. The human operator manipulates the master to move the end-effectors of all slaves. The task-space positions of all robot end-effectors and their inter-connections are displayed to the master for visual feedback to the operator. In all experiments, the user strives to move the master along the same path under the same driving force. The master-informed slave~$1$ connection is permanent, but slaves~$2$ and~$3$ can exchange information with slave~$1$ only when they are within $r=0.1$~m distance from it. Thus, the tree network of the slave swarm is $\mathcal{G}(0)=\{\mathcal{V},\mathcal{E}(0)\}$ with $\mathcal{V}=\{1,2,3\}$ and $\mathcal{E}(0)=\{(1,2),(1,3)\}$. Virtual point-based control~\cite{Lee2013TMECH} and the proposed dynamic control synchronize the slave swarm while preserving the interaction links $(1,2)$ and $(1,3)$ under the perturbation $\fbm{f}$. The video of all experiments is available at \url{https://youtu.be/UDAJAbRszS0}.

\begin{figure}[!hbt]
\centering
\includegraphics[width=\columnwidth]{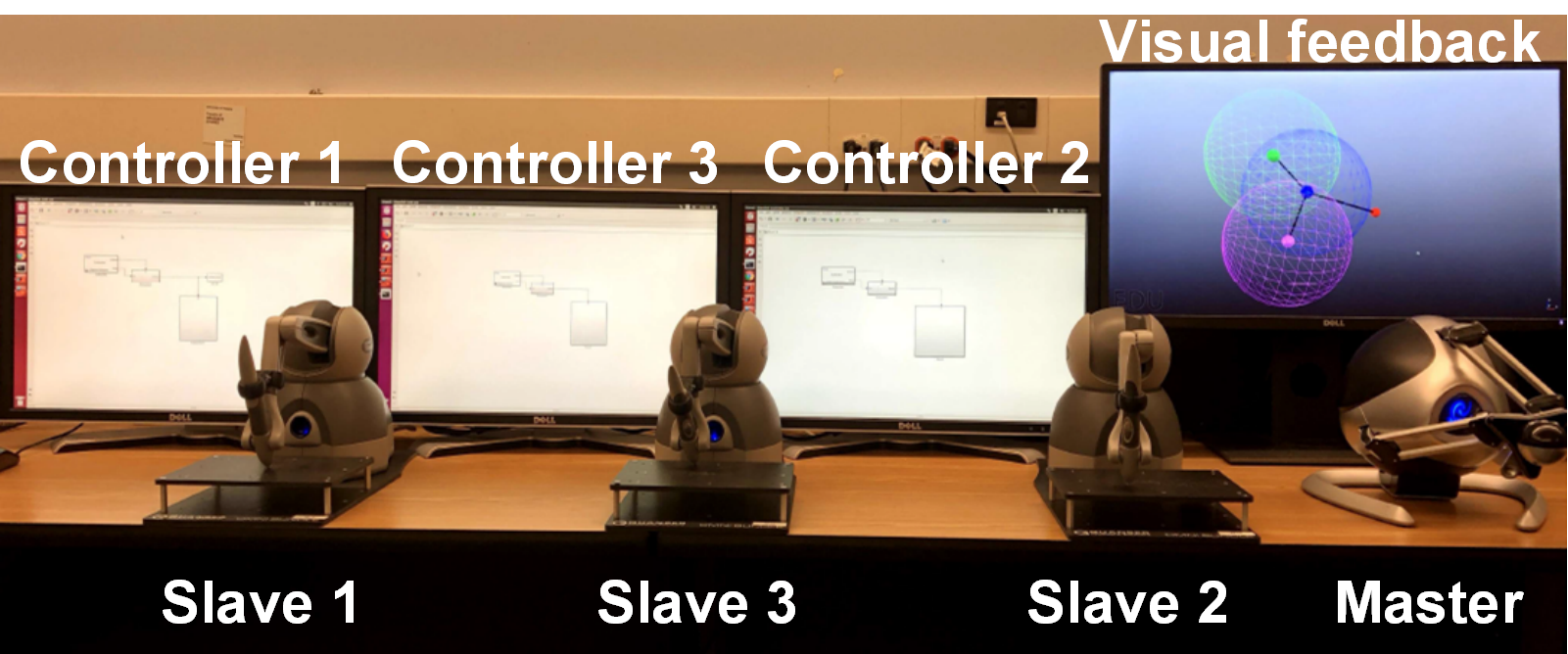}
\caption{The experimental swarm teleoperation testbed.}
\label{fig2}
\end{figure} 

%%%%%%%%
\noindent\textbf{Virtual point-based control~\cite{Lee2013TMECH}}

Because the robot parameters are unknown, each slave $i$ tracks its kinematic virtual point by Proportional plus damping control $\fbm{u}_{i}=K\cdot(\fbm{p}_{i}-\fbm{x}_{i})-d\dbm{x}_{i}$, where $\fbm{p}_{i}$ is the position of the virtual point, and $K=50$ and $d=3$ are heuristically tuned to optimize the tracking performance. Connectivity preservation control is implemented on the virtual point layer as in~\cite{Lee2013TMECH}. The virtual point of slave~$1$ is connected to the master through saturated Proportional control $\fbm{f}=\Sat\left(K_{0}\cdot(\fbm{x}_{0}-\fbm{p}_{1})\right)$, where $\fbm{x}_{0}$ is the position of the master robot and $K_{0}=10$ is selected to avoid destabilizing the experimental teleoperation.
\begin{figure}[!hbt]
\centering
\includegraphics[width=\columnwidth]{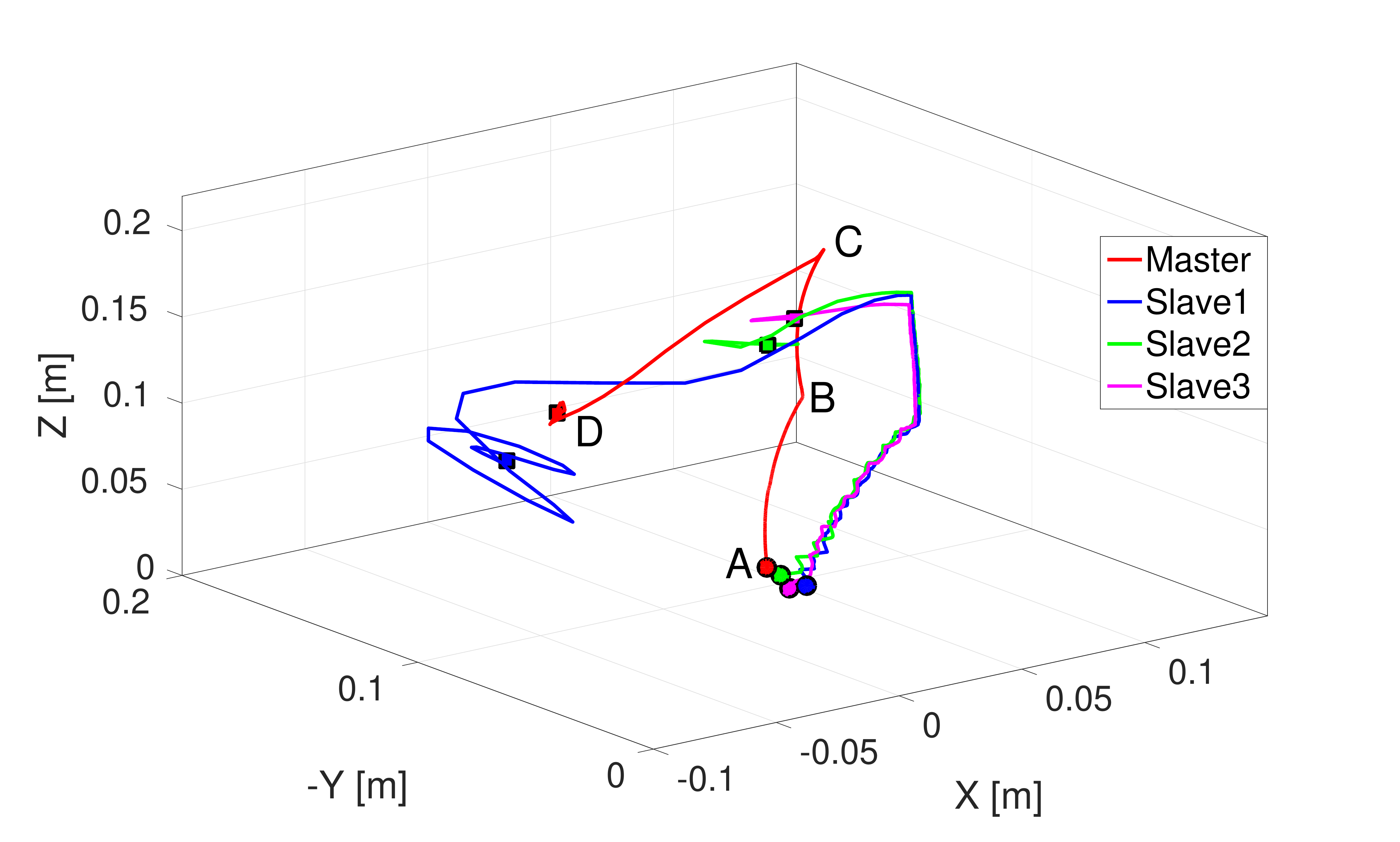}
\caption{The task-space paths of the master and all slave end-effectors under virtual point-based control~\cite{Lee2013TMECH}.}
\label{fig3}
\end{figure}

\fig{fig3} depicts the task-space paths of all robot end-effectors. Initially, the user moves the master slowly from A to B and then to C. The compliant connection $\fbm{f}$ permits the gravity compensation errors to cause relatively large position errors between the master and slaves. Nonetheless, all slaves are synchronized and follow the motion of the master. As \fig{fig4} depicts, the inter-slave distances remain smaller than $0.02$~m during the first $35$ seconds. Then, the user increases their perturbation by moving the master quickly from C to D. Because slaves~$2$ and~$3$ cannot track slave~$1$ closely enough, the inter-slave communication links are broken at about the $37.5$-th second, see~\fig{fig4}. Slaves~$2$ and~$3$ hence stop following the master while slave~$1$ moves towards D in~\fig{fig3}.
% in the second phase. As shown in~\fig{fig4}, the distances between slave $1$ and slaves $2$ and $3$ dramatically increase to over $0.1$~m at about the $37.5$-th second. Therefore, slaves $2$ and $3$ stop following the master while slave $1$ approaches D in~\fig{fig3}.
\begin{figure}[!hbt]
\centering
\includegraphics[width=\columnwidth]{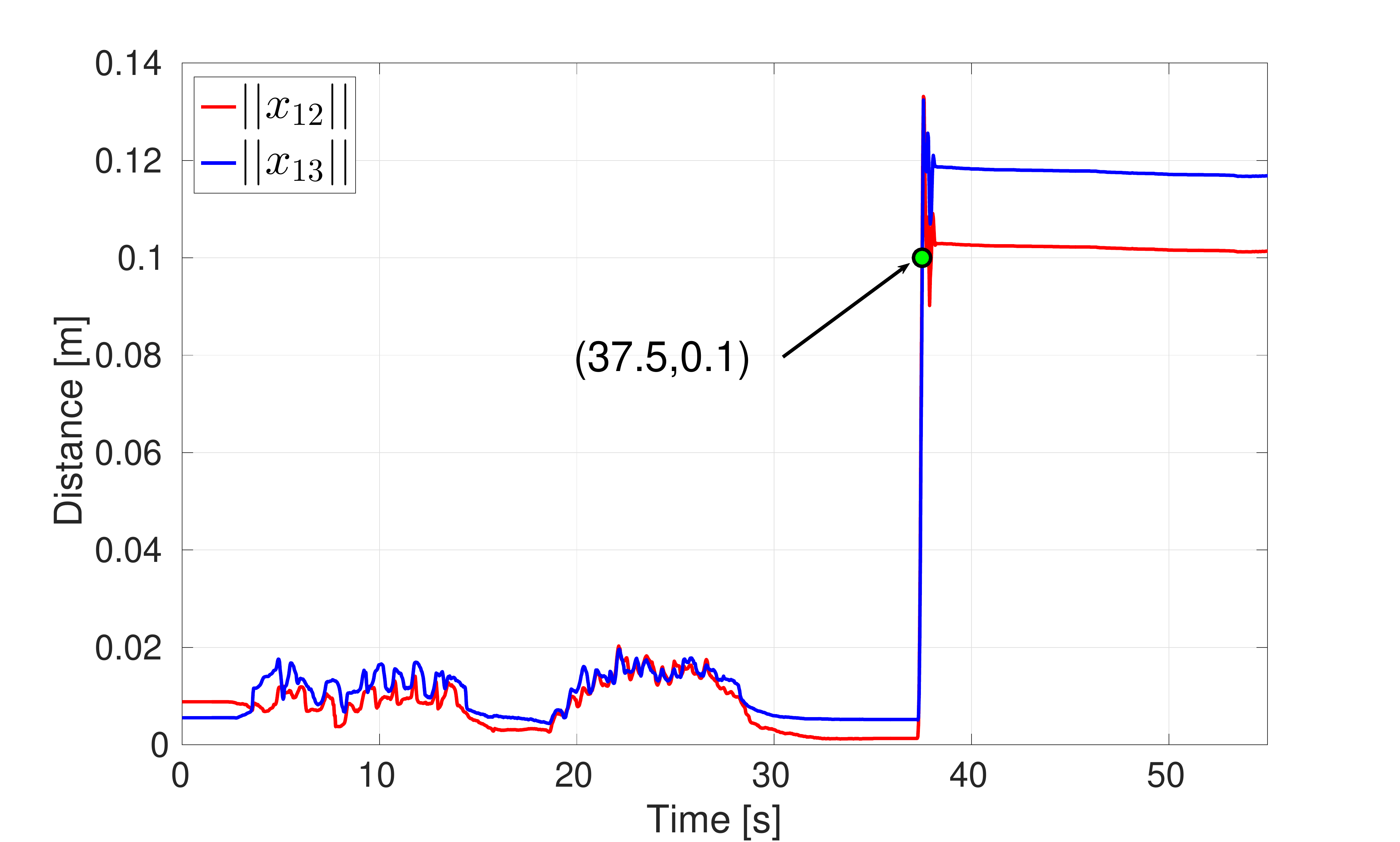}
\caption{The distances between slave $1$ and slave $2$ and $3$ under virtual point-based control~\cite{Lee2013TMECH}.}
\label{fig4}
\end{figure}

%%%%%%%%
\noindent\textbf{Dynamic control~\eqref{equ13} with no delay}

In this experiment, slave~$1$ is directly connected to the master by saturated Proportional control $\fbm{f}=\Sat\left(K_{0}\cdot(\fbm{x}_{0}-\fbm{x}_{1})\right)$ with the same gain $K_{0}=10$ as in virtual point-based control. After choosing $\rho=\sigma=B=1$ and $\eta_{i}=\gamma_{i}=\zeta_{i}=0.1$ heuristically, the damping gain $D=3$ makes $\hat{D}_{i}\geq 0$. The distances between adjacent slaves being initially smaller than $\hat{r}=0.02$~m, the selection $Q=1$ guarantees $\omega_{2}=0.0093>0$. Let $P=20$ by~\eqref{equ23} with $\Delta=0.1$ and update $K_{i}(t)$ by~\eqref{equ24} with estimated $\lambda_{i2}=0.01$ to make $\omega_{3}=0.0851>0$ and complete the controller design.
\begin{figure}[!hbt]
\centering
\includegraphics[width=\columnwidth]{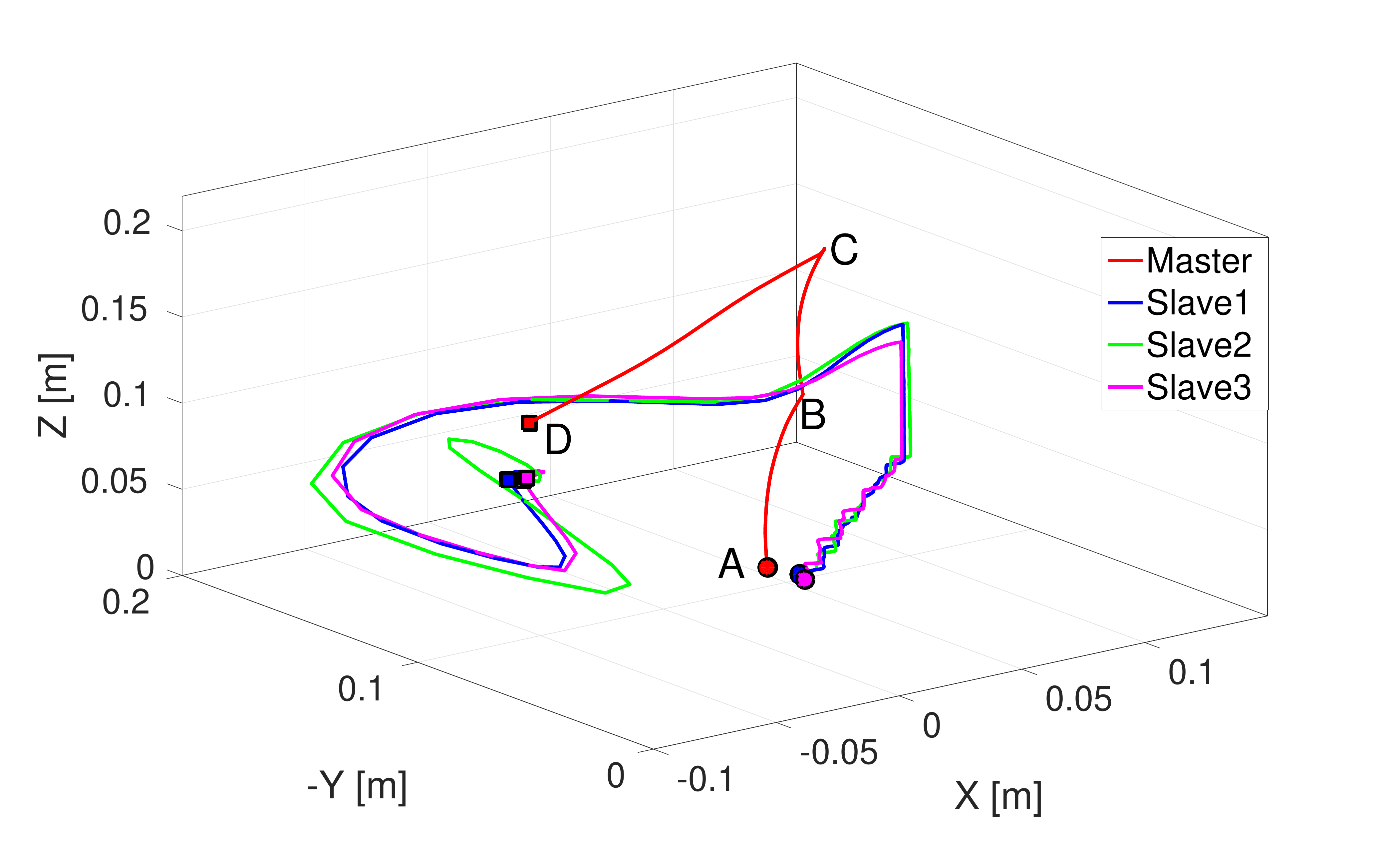}
\caption{The task-space paths of the master and all slave end-effectors under the dynamic control~\eqref{equ13} with no delay.}
\label{fig5}
\end{figure}

\fig{fig5} plots the task-space paths of all robot end-effectors under the dynamic control~\eqref{equ13}. Like virtual point-based control~\cite{Lee2013TMECH}, the dynamic control~\eqref{equ13} initially synchronizes all slaves and enables the user to guide the slave swarm by moving the master from A to C. Then, the user increases their perturbation $\fbm{f}$ during the movement from C to D. The swarm overshoots after the master stops, but all slaves asymptotically approach the master. \fig{fig6} shows that the inter-slave distances are smaller than $0.01$~m during the first $25$ seconds, and grow sharply to about $0.03$~m at about the $27$~s, when the user moves the master fast. Nonetheless, the tree connectivity is maintained because inter-slave distances remain much smaller than $r=0.1$~m.    
\begin{figure}[!hbt]
\centering
\includegraphics[width=\columnwidth]{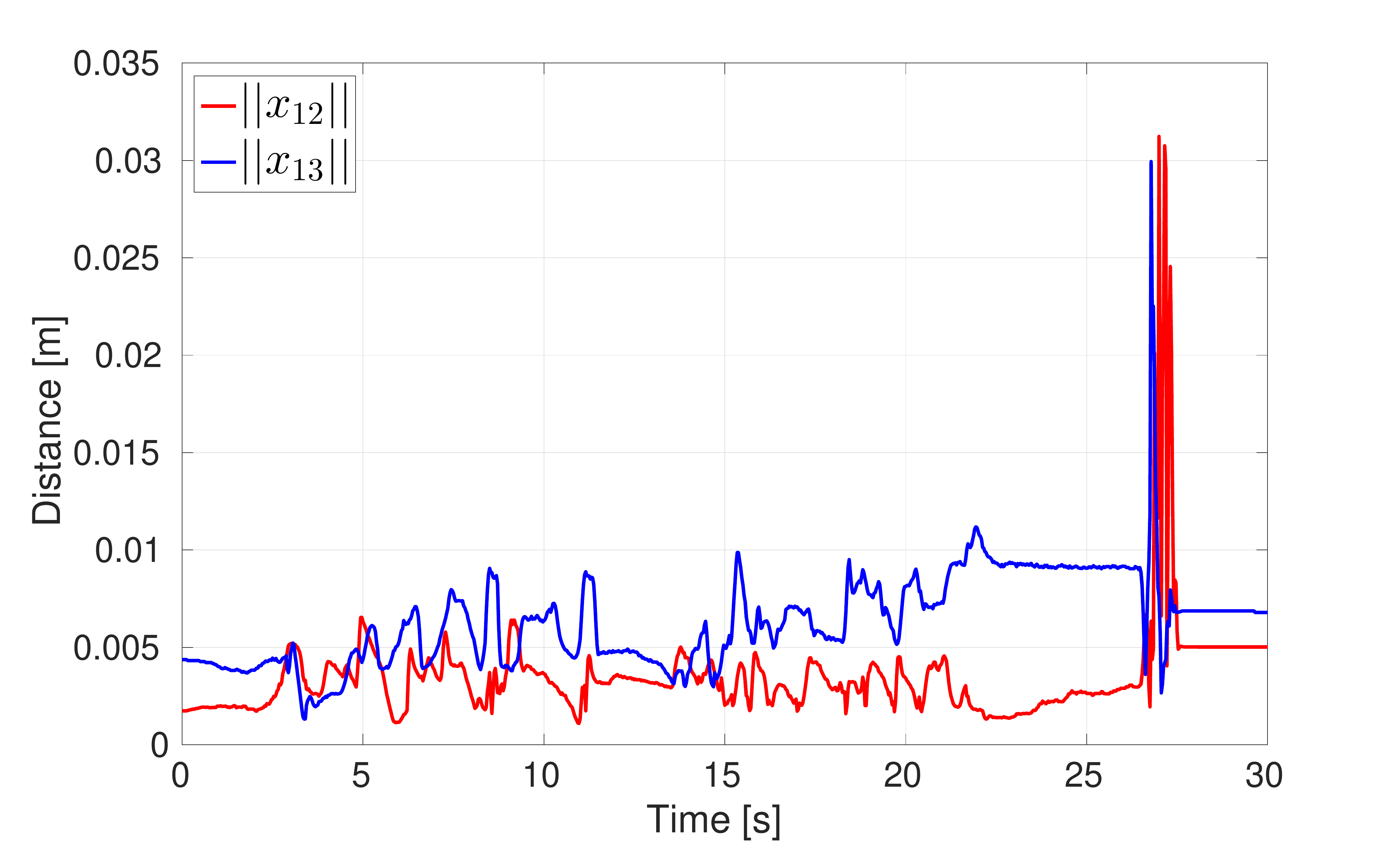}
\caption{The distances between slave $1$ and slave $2$ and $3$ under the dynamic control~\eqref{equ13} with no delay.}
\label{fig6}
\end{figure}

%%%%%%%%
\noindent\textbf{Dynamic control~\eqref{equ30} with time-varying delays}

For time-varying delays bounded by $\overline{T}=0.01$~s and given $\epsilon=0.08$~m and $\hat{r}=0.02$~m, the choice $\kappa=0.5$ makes $\overline{r}=0.06$~m and thus $\omega_{1}=2.944\times 10^{-5}>0$. After selecting $\overline{\sigma}=1$ and $B=1$ heuristically, the choice $Q=1$ makes $\omega_{2}=0.0028>0$. Because $\overline{\rho}_{i}\leq 177.15$, the selections $P=15$ and $\sigma=0.1268$ satisfy condition~\eqref{equ41} with $\overline{\Delta}=0.035$. Then let $\eta_{i}=\gamma_{i}=\zeta_{i}=0.1$ and update $K_{i}(t)$ by~\eqref{equ47}. Given estimated $\lambda_{i2}=0.01$ and $c_{i}=0.01$, it follows that $\overline{\Lambda}_{ij1}\leq 0.237$, $\overline{\Lambda}_{ij2}\leq 30.9$, $\overline{K}_{i}\leq 2.515$ and, by~\eqref{equ46}, that $\Upsilon=3696.1$. Sufficient damping $D=10$ guarantees $\widetilde{D}_{i}>0$ by~\eqref{equ40} and completes the controller design. 
\begin{figure}[!hbt]
\centering
\includegraphics[width=\columnwidth]{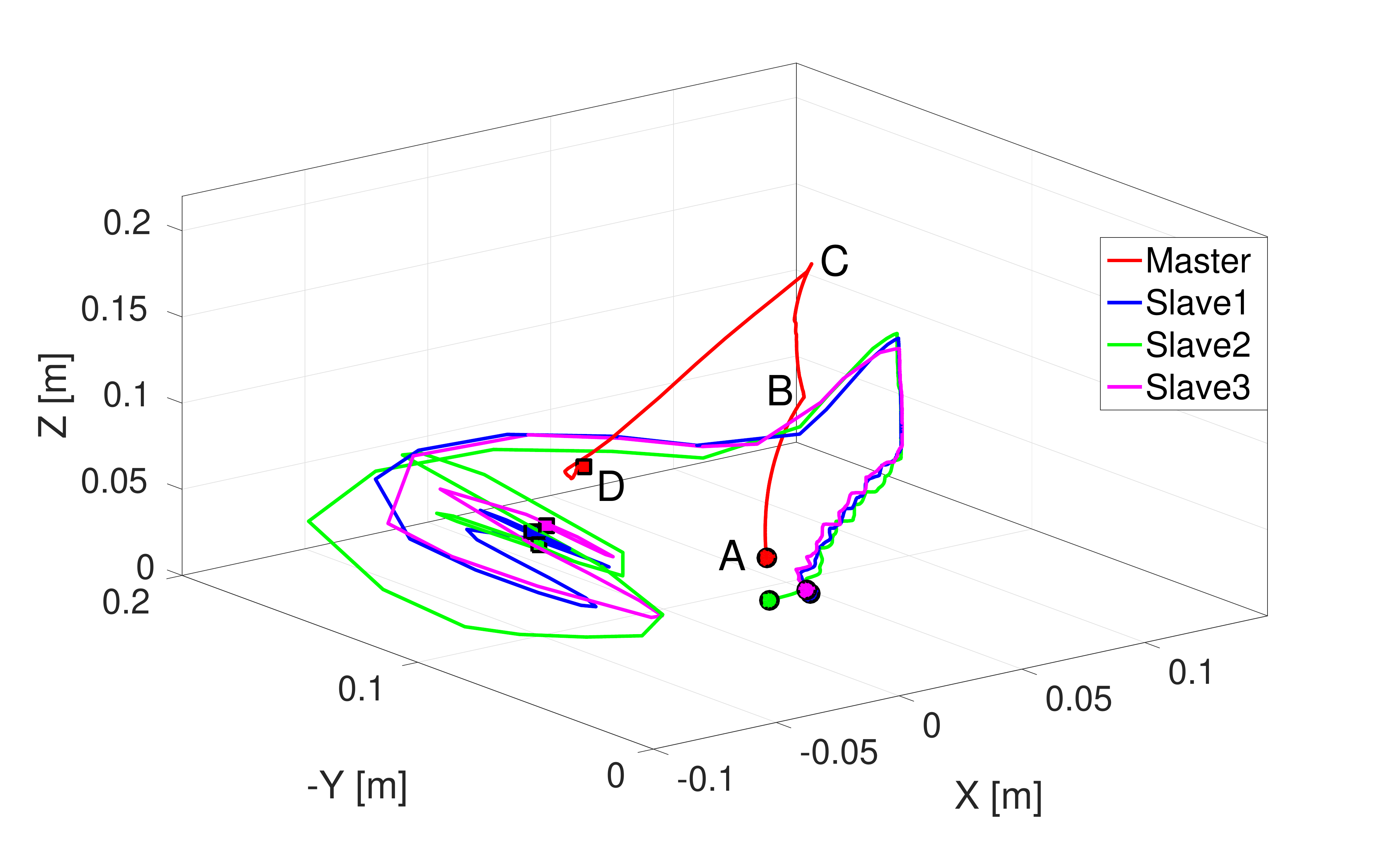}
\caption{The task-space paths of the master and all slave end-effectors under the dynamic control~\eqref{equ30} with time-varying delays up to $\overline{T}=0.01$~s.}
\label{fig7}
\end{figure}

\fig{fig7} shows the task-space paths of all robot end-effectors under the dynamic control~\eqref{equ30}. The impact of the time-varying delays is not very noticeable when the master moves from A to C slowly and all slaves are coordinated tightly. However, the delays disturb the inter-slave connections significantly when the user drives the master from C to D quickly. The slave swarm vibrates more than in \fig{fig5}, especially in the final phase of the movement. However, the dynamic control~\eqref{equ30} preserves all inter-slave connections and the swarm converges to the master.
\begin{figure}[!hbt]
\centering
\includegraphics[width=\columnwidth]{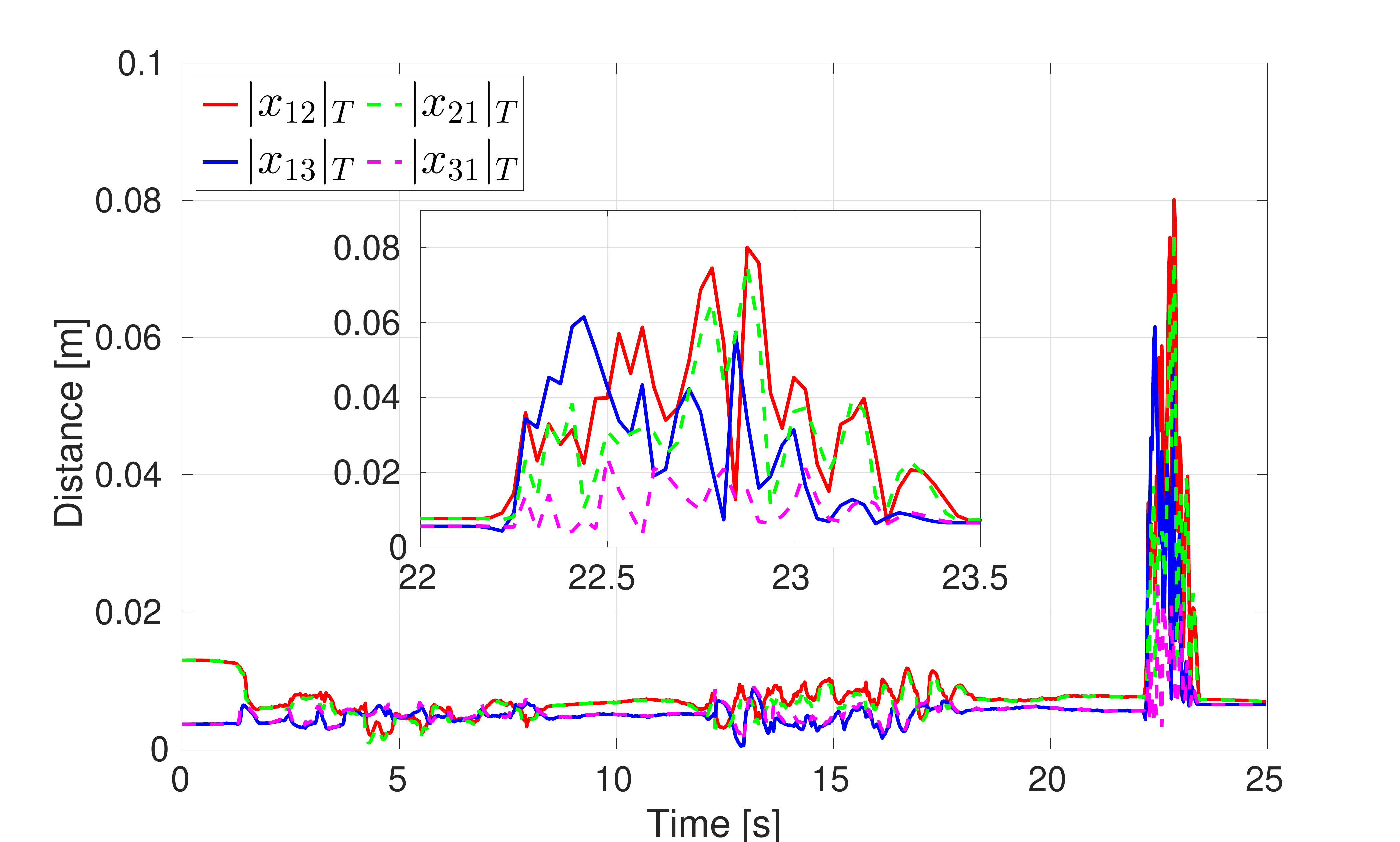}
\caption{The mismatched distances between slave $1$ and slaves $2$ and $3$ under the dynamic control~\eqref{equ30} with time-varying delays up to $\overline{T}=0.01$~s.}
\label{fig8}
\end{figure}

\fig{fig8} plots the mismatched inter-slave distances. They are strictly smaller than $0.02$~m during the first $20$ seconds when all slaves move slowly, and increase significantly at about the $23$-rd second, when the master moves fast. Nonetheless, the dynamic control~\eqref{equ30} guarantees that the initially adjacent slaves remain closer than $r=0.1$~m.  

%%%%%%%%%%%%%%%%%%%%%%%%%%%%%%%%%%%%%
\section{Conclusions}
%%%%%%%%%%%%%%%%%%%%%%%%%%%%%%%%%%%%%

This paper has proposed a constructive strategy that maintains the initial tree connectivity of a teleoperated swarm in the presence of unpredictable operator commands and time-delay inter-slave communications. The strategy modulates the inter-swarm couplings and the damping injected to each slave based on a customized potential. After sliding surfaces reduce the order of the swarm dynamics, this potential uses the inter-slave distances to quantify the impact of state-dependent mismatches on the connectivity of the swarm. Therefore, it serves to design an explicit gain updating law that suppresses the impact of the mismatches. Lyapunov-based set invariance analysis proves that the proposed explicit gain updating law limits the impact of the operator command and preserves the initial connectivity of a delay-free swarm. Further augmentation with stricter selection of control gains makes the design robust to time-varying delays in inter-slave transmissions. The paper also establishes the input-to-state stability of a teleoperated time-delay swarm under the proposed dynamic control. Future research will consider connectivity-preserving swarm teleoperation with limited actuation and heterogeneous communication radius.

\bibliography{IEEEabrv,Cooperation}
%\bibliography{bibi}
\end{document}